\documentclass[a4paper,10pt]{amsart}

\usepackage[latin1]{inputenc}

\usepackage[normalem]{ulem}
\usepackage{amssymb}
\usepackage[pdftex]{color}
\usepackage{hyperref}
\hypersetup{
	colorlinks   = true, 
	urlcolor     = brown, 
	linkcolor    = black, 
	citecolor   = black 
}
\usepackage{array}
\usepackage{mathtools}
\usepackage{multicol}
\usepackage[nameinlink, capitalise, noabbrev]{cleveref}
\usepackage{url}
\usepackage{amsmath}
\usepackage{graphicx,indentfirst,amsmath,amsfonts,amssymb,amsthm,newlfont}
\usepackage{epsfig}
\usepackage{enumerate}
\usepackage{tikzsymbols}
\usepackage{float}
\usepackage[font=footnotesize]{caption} 
\usepackage{subcaption}

\definecolor{forestgreen}{rgb}{0.13, 0.65, 0.13}
\usepackage[T1]{fontenc}
\usepackage{tikz}
\usetikzlibrary{positioning}
\usetikzlibrary{shapes.geometric, arrows}

\tikzstyle{blockS} = [rectangle, rounded corners, minimum width=1cm, minimum height=0.8cm,text centered, draw=black, fill=green!30]
\tikzstyle{blockE} = [rectangle, rounded corners, minimum width=1cm, minimum height=0.8cm,text centered, draw=black, fill=yellow!30]
\tikzstyle{blockI} = [rectangle, rounded corners, minimum width=1cm, minimum height=0.8cm,text centered, draw=black, fill=red!20]
\tikzstyle{blockR} = [rectangle, rounded corners, minimum width=1cm, minimum height=0.8cm,text centered, draw=black, fill=blue!30]
\tikzstyle{blockD} = [rectangle, rounded corners, minimum width=1cm, minimum height=0.8cm,text centered, draw=black, fill=black!30]
\tikzstyle{arrowflow} = [thick,->,>=stealth]
\tikzstyle{arrowinteraction} = [thick,->,>=stealth,dashed]

\newtheorem{theorem}{Theorem}[section]

\newtheorem{remark}[theorem]{Remark}
\newtheorem{definition}[theorem]{Definition}
\newtheorem{lemma}[theorem]{Lemma}

\newtheorem{hypothesis}[theorem]{Hypothesis}

\newcommand{\norm}[1]{\| #1 \|_{[0,T]}}

\newcommand{\maig}{\geqslant}

\newcommand{\R}{\mathbb R}

\newcommand{\E}{\mathbf E}
\newcommand{\F}{\mathcal F}
\newcommand{\A}{\mathcal A}
\newcommand{\B}{\mathcal B}
\newcommand{\K}{\mathcal K}

\newcommand{\M}{\mathcal M}

\newcommand{\diag}{\mathrm{diag}}

\newcommand{\orange}[1]{\textcolor{black}{#1}}

\everymath{\displaystyle}

\begin{document}

\title[Final size and partial distance estimate for the SEIRD model]{Final size and partial distance estimate for a two-group SEIRD model \thanks{Submitted on \today}}

\author{Alison M.V.D.L. Melo}
\address{A.M.V.D.L.M.: Universidade Federal do Vale do São Francisco - UNIVASF, 56304-917, Petrolina, Brazil}
\email{alison.melo@univasf.edu.br}

\author{Matheus C. Santos}
\address{M.C.S.: Departamento de Matem\'atica Pura e Aplicada-- IME, Universidade Federal do Rio Grande do Sul - UFRGS, 91509-900, Porto Alegre, Brazil}
\email{matheus.santos@ufrgs.br}

\date{June 23, 2021}

\begin{abstract}
		In this paper we consider a SEIRD epidemic model for a population composed by two groups of individuals with asymmetric interaction.  Given an \orange{approximate} solution for the two-group model, we estimate the error of this approximation to the unknown solution to the second group based on the known error that the approximation has with respect to the solution to the first group. We also study the final size of the epidemic for each group. We illustrate our results with the spread of the coronavirus disease 2019 (COVID-19) pandemic in the New York County (USA) for the initial stage of the contamination, and in the cities of Petrolina and Juazeiro (Brazil).
\end{abstract}



\maketitle

\section{Introduction}

The models for outbreak and spread of diseases have a long history of studies and analysis since the pioneering work of John Graunt, who made {a statistical} approach to understand public health and causes of deaths in England (\cite{Graunt}), and Daniel Bernoulli, who proposed the first mathematical model describing an infectious disease (\cite{Bernoulli}) and the impact of inoculation for the smallpox control (\cite{Bernoulli2}). The models we use nowadays (SIR models and extensions) originated in the works of Sir Ronald Ross (\cite{Ross}), who formulated a system of differential equations after his studies on the spread of malaria, and later, refined by William Kermack and Anderson McKendrick, who generalized the approach and established the condition for an epidemic to occur (\cite{KMcK}), which is now known as the basic reproduction number. Many related models and approaches have been studied since then (\cite{Brauer,Keeling,Maia,Murray}), and were specially used \orange{recently} to analyze the spread and to forecast the number of cases of COVID-19 epidemic (\cite{Bertuzzo1,Bertuzzo2,GM,MagalLatency,LMS,Prodanov,RLLKRHYC,SCZ,TT} and many others).

In this paper we study the case of a two-group epidemic model, where the infectious individuals of each group transmit the disease to the susceptible individuals of both groups according to an asymmetric interaction. The division in groups may happen for biological, geographic or even socioeconomic reasons. For instance, in \cite{ACEDO} \orange{an age structured} SIRS model was applied to study seasonal evolution of Respiratory Syncytial Virus (RSV) in \orange{Valencia}, Spain. Infection by RSV tends to be more severe \orange{in babies} under one year old. On the other hand, COVID-19 affects elderly people more aggressively  (\cite{NEJ,WHO}). A disease may also spread differently between locations. In \cite{CHALVET} \orange{the authors} study the spread of sleeping sickness by using a differential equation model where \orange{both human and vector populations} are divided into two patches according to location: plantation and village. Finally, socioeconomic inequality might play an important role in the spread of diseases such as COVID-19 as shown in \cite{SEAN} and \cite{MARTINS}.

After a susceptible individual becomes infected by interacting with an infectious one, there is a period of time before he can transmit the disease himself. This period of latency or exposure may or may not be relevant for the dynamics of a disease, depending on how long it may be \orange{compared to} the total duration of the infectious period. In the case of COVID-19, the latency period is on average 3-5 days (\cite{Lauer,Wang}) up to 14 days, which is 1-2 days less than the  incubation period, i.e., the period needed for the symptoms onset, which is on average 5-6 days, also up to 14 days (\cite{WHO}). Therefore, since the infectious period may vary between 5 days (mild cases) up to 6 weeks (severe or critical cases) (\cite{WHO,Lin}), this exposure period is relevant to the progression of the epidemic. That is why in this work we have chosen to study the two group SEIRD model, which includes the exposure period.

In this article, we study the system of equations for the two-group SEIRD model. The
system considered here is the following:
\begin{equation}\label{Eq.SEIRD}
\left\{ \begin{array}{l}
S'(t) = -\diag(S(t))\B I(t) \\
E'(t) = \diag(S(t))\B I(t) - \A E(t)\\
I'(t) = \A E(t) - (\Gamma+\M)I(t)\\
R'(t) = \Gamma I(t)\\
D'(t) = \M I(t)
\end{array}
\right.
\end{equation}
with initial conditions
\[S(0)=S_0 ,\; E(0)=E_0,\; I(0)=I_0,\; R(0)=R_0,\;D(0)=D_0\in \R^2_+, \]
where the vectorial functions $S(t),E(t),I(t),R(t)$ and $D(t)$ represent the classes of susceptible, exposed, infectious, recovered and deceased individuals, respectively, for groups 1 and 2 as
\[ S(t)\!=\!\begin{pmatrix}
S_1(t)\\S_2(t)
\end{pmatrix} ,
E(t)\!=\!\begin{pmatrix}
E_1(t)\\E_2(t)
\end{pmatrix} ,
I(t)\!=\!\begin{pmatrix}
I_1(t)\\I_2(t)
\end{pmatrix} ,
R(t)\!=\!\begin{pmatrix}
R_1(t)\\R_2(t)
\end{pmatrix} ,
D(t)\!=\!\begin{pmatrix}
D_1(t)\\D_2(t)
\end{pmatrix} .
\]

Note that, if $N=(N_1,N_2)$ is the number of individuals at each group at time $t=0$, by the system of equations \eqref{Eq.SEIRD} we have
\begin{equation}\label{Eq.NumIndividuals}
S(t)+E(t)+I(t)+R(t)+D(t) = N\;,\;\;\; \mbox{ for all }\; t\maig 0.
\end{equation}

The transmission of the disease is given by the interaction of the infectious individuals with the susceptible ones and expressed by the infection matrix, composed by the infection rates $\beta_{ij}$ of the susceptible individuals $S_i$ by the infectious individuals $I_j$
\[   \B = \begin{pmatrix}
\beta_{11}& \beta_{12} \\
\beta_{21}& \beta_{22}
\end{pmatrix} .\]
Once infected, each individual {becomes} an asymptomatic noninfectious individual in class $E$. The exit of the exposed class is given by the latency matrix
\[\A = \left( \begin{array}{cc}
\alpha_1& 0 \\
0& \alpha_2
\end{array}\right),\]
where $\alpha_1,\alpha_2>0$. This means that each newly infected individual takes the average time of $1/\alpha_i$ (according to the group) to become infectious, i.e, to be in the class $I$ and contribute with the communicability of the disease. The dynamics of individuals leaving the class $I$ is governed by the recovering rate matrix $\Gamma$ and the death rate matrix $\M$ given by
\[
\Gamma = \left( \begin{array}{cc}
\gamma_1& 0 \\
0& \gamma_2
\end{array}\right) \;,\;\; 
\M = \left( \begin{array}{cc}
\mu_1& 0 \\
0& \mu_2
\end{array}\right) .
\]
Thus the exit flux of infectious individuals {in each group} is composed by $\gamma_i I_i$, which represents those who recovered from the disease and become immune (per time unit), and $\mu_i I_i$ which gives the rate of deaths {caused} by the disease per unit of time. Hence the average time of infectiousness in each group is $1/(\gamma_i+\mu_i)$.

The flux of individuals for system \eqref{Eq.SEIRD} is shown in Figure \ref{Fig.flow}.

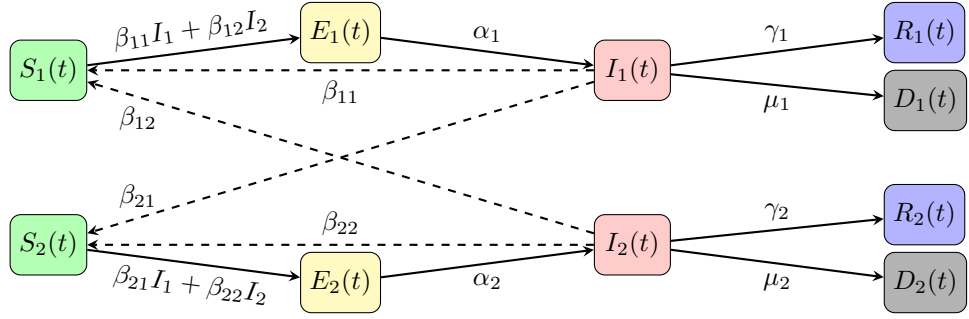
\begin{figure}[H]
	\centering
	\begin{tikzpicture}
	\tikzset{node distance = 1.5cm and 2.8cm}
	\node (S1) [blockS] {$S_1(t)$};
	\node (S2) [blockS, below= of S1] {$S_2(t)$};
	
	\node (E1) [blockE, right= of S1,yshift=0.5cm] {$E_1(t)$};
	\node (E2) [blockE, right= of S2,yshift=-0.5cm] {$E_2(t)$};

	\node (I1) [blockI,right = of E1,yshift=-0.5cm] {$I_1(t)$};
	\node (I2) [blockI, below = of I1] {$I_2(t)$};

	\node (R1) [blockR,right = of I1,yshift=0.5cm] {$R_1(t)$};
	\node (R2) [blockR, right= of I2,yshift=0.4cm] {$R_2(t)$};

	\node (D1) [blockD, right = of I1,yshift=-0.4cm] {$D_1(t)$};
	\node (D2) [blockD, right = of I2,yshift=-0.5cm] {$D_2(t)$};

	\draw [arrowflow] (S1) -- node[sloped,anchor=south] {$\beta_{11}I_1+\beta_{12}I_2$} (E1);
	\draw [arrowflow] (E1) -- node[anchor=south] {$\alpha_1$} (I1);
	\draw [arrowflow] (I1) -- node[anchor=south] {$\gamma_1$} (R1);
	\draw [arrowflow] (I1) -- node[anchor=north] {$\mu_1$} (D1);
	
	\draw [arrowflow] (S2) -- node[sloped,anchor=north] {$\beta_{21}I_1+\beta_{22}I_2$} (E2);
	\draw [arrowflow] (E2) -- node[anchor=north] {$\alpha_2$} (I2);
	\draw [arrowflow] (I2) -- node[anchor=south] {$\gamma_2$} (R2);
	\draw [arrowflow] (I2) -- node[anchor=north] {$\mu_2$} (D2);

	\draw [arrowinteraction] (I1) to node[pos=0.9, anchor=south] {$\beta_{21}$} (S2);
	\draw [arrowinteraction] (I2) to node[pos=0.9, anchor=north] {$\beta_{12}$}  (S1);
	\draw [arrowinteraction] (I1) [bend left=0] to node[pos=0.5,anchor=north] {$\beta_{11}$} (S1);
	\draw [arrowinteraction] (I2) [bend right=0] to node[pos=0.5,anchor=south] {$\beta_{22}$}(S2);
	
	\end{tikzpicture}
	\caption{Flow diagram of individuals in system \eqref{Eq.SEIRDcomponentwise}. The solid arrows represents the flux of individuals and the dashed arrows represent the interaction of infectious classes $I_1$ and $I_2$.}
	\label{Fig.flow}
\end{figure}

Therefore the system \eqref{Eq.SEIRD} can be explicitly rewritten as
\begin{equation}\label{Eq.SEIRDcomponentwise}
\left\{ \begin{array}{lll}
S_1'  = -(\beta_{11}I_1 +\beta_{12}I_2 ) S_1    &,& S_2'  = -(\beta_{21}I_1 +\beta_{22}I_2 ) S_2 \\
E_1'  = (\beta_{11}I_1 +\beta_{12}I_2 ) S_1 - \alpha_1 E_1  \qquad&,& E_2'  = (\beta_{21}I_1 +\beta_{22}I_2 ) S_2  - \alpha_2 E_2  \\
I_1'  = \alpha_1 E_1  - (\gamma_1+\mu_1)I_1  &,&I_2'  = \alpha_2 E_2  - (\gamma_2+\mu_2)I_2 \\
R_1'  = \gamma_1 I_1 &,&R_2'  = \gamma_2 I_2 \\
D_1'  = \mu_1 I_1 &,&D_2'  = \mu_2 I_2 
\end{array}\right.
\end{equation}

The present work deals with two different problems concerning the two-group SEIRD model. The first problem is estimating the final size of the epidemic.  {The study of the final size SIR models with multi-group population  was made in \cite{Andreasen} and \cite{Magal1}, for example.} We discuss it in  {Section} \ref{Section.FinalSize} and prove that for both groups the number of individuals who escape the {epidemic} is always positive. {We} also present a way of calculating these limit values in Theorem \ref{Theo.FinalSize}.

 {For the second problem, consider two sets of parameters, $\widetilde{\A}$, $\widetilde{\B}$, $\widetilde{\Gamma}$, $\widetilde{\M}$ and $\A, \B $, $\Gamma$, $\M$, for the system \eqref{Eq.SEIRD} with the same initial conditions, and let $\widetilde{S}, \widetilde{E}, \widetilde{I}, \widetilde{R}, \widetilde{D}$ and $S, E, I, R, D$ be their respective sets of solutions. We would like to obtain estimates for the distance between the solutions to second group, i.e., for the quantities $\norm{\widetilde{S}_2-S_2}$, $\norm{\widetilde{E}_2-E_2}$, $\norm{\widetilde{I}_2-I_2}$, $\norm{\widetilde{R}_2-R_2}$ and $\norm{\widetilde{D}_2-D_2}$, based on the respective distances between the solutions to first group. In other words, if the respective solutions to the first group are close to each other, than we would like to know how close the respective solutions to the second group are. We can also think of this problem as follows: suppose that the parameters $\widetilde{\A}$, $\widetilde{\B}$, $\widetilde{\Gamma}$, $\widetilde{\M}$ and the solutions to the second group $\widetilde{S}_2, \widetilde{E}_2, \widetilde{I}_2, \widetilde{R}_2, \widetilde{D}_2$ are unknown, but the solutions to the first group are known. Also, suppose that the parameters $\A, \B $, $\Gamma$, $\M$ and the respective solutions for both groups are known. Then, from the error $\norm{(\widetilde{S}_1, \widetilde{E}_1, \widetilde{I}_1, \widetilde{R}_1, \widetilde{D}_1)-(S_1, E_1, I_1, R_1, D_1)}$, we would like to estimate the respective error for the solutions to the second group, i.e., we would like to specify a range around the solutions $S_2, E_2, I_2, R_2, D_2$ where the unknown $\widetilde{S}_2, \widetilde{E}_2, \widetilde{I}_2, \widetilde{R}_2, \widetilde{D}_2$ must be. Unfortunately, we were not able to provide such estimates in this full generality. However, in Section \ref{Section.Estimates}, we present a way to solve this problem under the additional hypothesis that the distance $\norm{\widetilde{D}_2-D_2}$ of number of deaths in group 2 is known. Thus, in Theorem \ref{Theo.Estim} we present these estimates under this extra information on group 2 and with some hypotheses on the parameters of the system. }


Finally, in Section \ref{Section.COVID} we illustrate these results with the evolution of COVID-19 epidemic in two two-group populations: New York County and its neighboring counties in USA in Subsection \ref{SubSection.NY}, and the neighboring cities of Petrolina and Juazeiro in Brazil, in Subsection \ref{SubSection.PetrJua}. In the first example, we are able to estimate the evolution of infectious cases in New York County only using its data of deaths and the evolution of cases and deaths in the neighboring counties. On the other hand, the same method applied to the second example of the cities in Brazil, suggests that the cases in Juazeiro may be under-reported, based on the evolution of cases and deaths in Petrolina and on the deaths in Juazeiro.

Estimating the distance for the solutions in the second group as mentioned above might be useful in situations where reporting on part of the cases of infection is compromised. This could happen, for instance, to an unprivileged minority or to poor people who might have no proper access to the health system, or even a structural feature of the health system, where a relative under-development of the health system in one group could lead to a lower number of reported cases.

\orange{ It is worth mentioning that we are assuming that infections by infected individuals from group $j$ are proportional to $I_j/N_j$, since the rate in which individuals in group $i$ become infected, as presented in \eqref{Eq.SEIRDcomponentwise}, can be rewritten as $\beta_{ij}I_j = \beta_{ij}N_j (I_j/N_j)$. For many infections that can cause deaths, this is not a realistic assumption since it does not take into account the change in number of alive individuals. A more accurate model should exclude the total number $D_j$ of deceased individuals from the infecting process and thereby suppose that the force of infection is proportional $I_j/(N_j-D_j)$.} For many applications, this is not a significant change if the number of deceased individuals is small relatively to the population size, as it is at the beginning of an epidemic for example. But in a long-time analysis, these two models may be considerably different, specially in cases where the mortality rates $\mu_i$ are large. In the following section we present results on the final size analysis not considering the effects that the change in alive population size has on the infectious rates, and the results presented does not seem adaptable to that case. Although this is a very interesting problem, the authors are not aware of any final size analysis that includes deaths apart from \cite{Diekmann}.

We are also assuming with model \eqref{Eq.SEIRDcomponentwise} that individuals from different groups interact freely in a homogeneous space, as studied in \cite{Andreasen} and \cite{Magal1}. The probability of an infected individual from group $i$ entering into contact with individuals from group $j$ is proportional to the group size. For a much more general approach that takes into account different population nodes and the flux of individuals between them, see for example \cite{Bertuzzo1} and \cite{Bertuzzo2}.

\orange{Finally, another limitation of model \eqref{Eq.SEIRD} that we would like to mention is that it considers $1/(\gamma_i+\mu_i)$ as the average time of infection for each group, whether it ends up with death or recovery.  A more precise model should take into account different periods of time for infection depending on the severity and the outcome of the cases, since, as we mentioned before, it may vary from 5 days to 6 weeks in the case of the COVID-19 disease. In \cite{SOFONEA}, the authors present a more accurate model that considers extra compartments that covers different age of infection depending on the severity and also on the outcome of the disease.}


\section{Final Size} \label{Section.FinalSize}
In this section we analyze the final size of the epidemic for the two-group model. We prove that the number of exposed and infectious individuals always goes to zero as time goes to infinity, and that there is always a positive number of individuals who escape the epidemic. We remark that our results for the final size of the two-group SEIRD model  {and the respective proof} are very similar to the results given in \cite{Magal1} for the SIR model, since the models are close. Therefore, our proof is adapted from the one given by the authors in that reference. The main difference here is the inclusion of the class of exposed individuals, which affects \orange{the time in which} an individual starts to contribute to the spread of the disease. 

{For the results in this section, we will need the following definition: }
\begin{definition}\label{def.irreducible}
{The $n\times n$ matrix $A$ is irreducible when $A$ cannot be transformed into block upper-triangular form by simultaneous row-column permutations, i.e, when there is no permutation matrix $P$ such that $PAP^{-1}$ is block upper-triangular.}

{Note that, if $n=2$, $A$ is irreducible if and only if $A_{12}, A_{21} \neq 0$.}
\end{definition}

We will assume the following hypotheses on the matrices $\B$ and $\Gamma+\M$:
\begin{hypothesis}\label{H1-B}
	$\B$ is a nonnegative irreducible matrix. {By Definition \ref{def.irreducible}, this is equivalent to assuming that $\beta_{12},\beta_{21}>0$.}
\end{hypothesis}

\begin{hypothesis}\label{H2-MN}
	\orange{$\alpha_1,\alpha_2>0$} ,  $\gamma_1, \gamma_2, \mu_1, \mu_2\geq0$ and $\gamma_1 +\mu_1>0$ and $\gamma_2 +\mu_2>0$
\end{hypothesis}

We start with the following representation formulas for the solutions $E_i$ and $I_i$ of the exposed and infectious individuals, respectively, of each group:

\begin{lemma}\label{explicitLemma} For each $i=1,2$, let  $S_i$, $E_i$ and $I_i$ be the components of the solution to the system \eqref{Eq.SEIRDcomponentwise}. Then
	\begin{eqnarray}
	E_i(t)&=&-S_i(t)+\alpha_i e^{-\alpha_i t}\int_0^tS_i(s)e^{\alpha_i s}\,ds +(S_i(0)+E_i(0))e^{-\alpha_i t}, \label{Eq.explicitE} \\
	I_i(t) &=& I_i(0)e^{-(\gamma_i+\mu_i) t}+\alpha_i e^{-(\gamma_i+\mu_i) t}\int_0^t e^{(\gamma_i+\mu_i)s}E_i(s)\,ds  . \label{Eq.explicitI}
	\end{eqnarray}
\end{lemma}

\begin{proof}
	Firstly, adding the equations for $S$ and $E$, we obtain
	\[   S'(t)+E'(t) = -\A E(t).     \]
	Componentwise, it holds for each $i=1,2$ that
	\begin{equation}\label{Eq.EeS}
	E_i'(t)+\alpha_i E_i(t) = -S'_i(t) .  
	\end{equation}   
	
	Thus, multiplying by $e^{\alpha_i t}$ we can rewrite it as
	\[  \frac{d}{dt}\left(e^{\alpha_i t}E(t)\right) = -e^{\alpha_i t} S'_i(t) .\]
	
	Finally, an integration by parts leads to the result.
	
	The expression for $I_i$ follows straight from the $I$-equation.
\end{proof}

With Lemma \ref{explicitLemma} above, we can obtain the asymptotic values for the solutions $E$ and $I$: 
\begin{lemma}
	For any initial conditions, there exist the limits 
	\begin{eqnarray*}
		S^\infty &=& \lim_{t\to+\infty}S(t) \;,\;\; E^\infty = \lim_{t\to+\infty}E(t)\;,\;\; I^\infty = \lim_{t\to+\infty}I(t)\;,\\  
		R^\infty &=& \lim_{t\to+\infty}R(t)\;\;\mbox{ and } \;\;D^\infty = \lim_{t\to+\infty}D(t) .
	\end{eqnarray*}
	Furthermore, by Lemma \ref{explicitLemma}, it follows that $E^\infty=0$ and $I^\infty =0$.
\end{lemma}

\begin{proof}
	Since $\beta_{i,j}\maig 0$, $\gamma_i>0$ and $\mu_i>0$ for all $i,j \in\{1,2\}$, it follows from the system \eqref{Eq.SEIRD} that $S_i(t), R_i(t)$ and $D_i(t)$ are monotone. Furthermore, by the restriction \eqref{Eq.NumIndividuals}, it follows that these functions are also bounded in $[0,\infty)$. Therefore, there exists the limits $S_i^\infty$, $R_i^\infty$ and $D_i^\infty$ for all $i=1,2$.
	
	For the functions $E_i$ and $I_i$, $i=1,2$, we use the representation given by Lemma \ref{explicitLemma}. Then applying L'H\^opital's Rule, we obtain that $E_i^\infty,I_i^\infty = 0$ for $i=1,2$.
\end{proof}

\begin{remark}
	Another way to see that $E^\infty=0$ is estimating
	\begin{eqnarray*}
		\alpha_ie^{-\alpha_i t}\int_{0}^{t} S_i(s)e^{\alpha_i s}\,ds &=& \alpha_ie^{-\alpha_i t}\int_{0}^{t/2} S_i(s)e^{\alpha_i s}\,ds+\alpha_ie^{-\alpha_i t}\int_{t/2}^{t} S_i(s)e^{\alpha_i s}\,ds\\
		&\leq& S_i(0)e^{-\alpha_i t}(e^{\alpha_i t/2}-1) +S_i(t/2)e^{-\alpha_i t}(e^{\alpha_i t}-e^{\alpha_i t/2})\\
		&=&\left(S_i(0)-S_i(t/2)\right)e^{-\alpha_i t/2}+S_i(0)e^{-\alpha_i t}+S_i\left(t/2\right) .
	\end{eqnarray*}
	Therefore
	\[0\leq E_i(t) \leq S_i(t/2)-S_i(t)+ \left(S_i(0)-S_i(t/2)\right)e^{-\alpha_i t/2} +E_i(0)e^{-\alpha_i t}\]
	and the result follows. This also shows that $S(t)\to S^\infty$ no faster than $E(t)\to 0$.
\end{remark}

In the rest of this section we show that the asymptotic values $S_1^\infty$ and $S_2^\infty$ are always positive in the two-group model and we will obtain a formula to estimate these values.

Let us begin with the equation for $S(t)$ in \eqref{Eq.SEIRD}. {Since $S_1(0),S_2(0)>0$, by continuity we can take $\ln S(t)$ at least for $t$ sufficiently small. Thus, we have that}
\[   \frac{d}{dt}(\ln S(t)) = -\B I(t)  \]
where the logarithm is taken componentwise. By adding the first three equations we obtain
\[  \frac{d}{dt}(S(t)+E(t)+I(t)) = -(\Gamma +\M) I(t) . \]
Therefore,
\[\frac{d}{dt}\Big( \B(\Gamma+\M)^{-1} (S(t)+E(t)+I(t)) - \ln S(t) \Big) = 0 .\]
Thus, we obtain that the function $\F$ defined by
\begin{equation}\label{Eq.Fconstante}
\F(t) := \B(\Gamma+\M)^{-1} (S(t)+E(t)+I(t)) - \ln S(t)
\end{equation}
{must be constant equal to $\F(0)$. By continuity, the logarithm above can be taken for all values of $t$, and the fact that $S^\infty$ is finite implies that we can take the limit $t\to +\infty$ and conclude that $S_1^\infty,S_2^\infty > 0$. Taking the limit $t\to\infty$ we also have the relation}
\begin{equation}\label{Eq.Sinfity}
\B(\Gamma+\M)^{-1} S^\infty - \ln S^\infty = \B(\Gamma+\M)^{-1} (S_0+E_0+I_0) - \ln S_0 ,   
\end{equation}
since $I^\infty = 0 = E^\infty$. Denoting  $X_0 := S_0+E_0+I_0$, we can write
\begin{eqnarray}
\B(\Gamma+\M)^{-1} S^\infty - \ln S^\infty &=& \B(\Gamma+\M)^{-1}X_0 - \ln S_0\nonumber \\        
\ln S^\infty &=&\ln S_0 +\B(\Gamma+\M)^{-1}( S^\infty - X_0)\nonumber \\
S^\infty &=&S_0\exp\Big(\B(\Gamma+\M)^{-1}( S^\infty - X_0)\Big), \label{Eq.Sinfty}
\end{eqnarray}
where the exponential is taken componentwise. Equation \eqref{Eq.Sinfty} shows that $S^\infty$ is a fixed point of the map $T:\R^2\to \R^2$ defined by
\begin{equation}\label{def.T}
	T(X)=S_0\exp\Big(\B(\Gamma+\M)^{-1}( X - X_0)\Big)
\end{equation}
or, denoting $X=(x,y)$ and $T(X)=(T_1(x,y),T_2(x,y))$,
\begin{eqnarray}
	T_1(x,y) &=& S_1(0)\exp\left(\frac{\beta_{11}}{\gamma_1+\mu_1}( x - x_0)+\frac{\beta_{12}}{\gamma_2+\mu_2}(y-y_0)\right), \\
	T_2(x,y) &=& S_2(0)\exp\left(\frac{\beta_{21}}{\gamma_1+\mu_1}( x - x_0)+\frac{\beta_{22}}{\gamma_2+\mu_2}(y-y_0)\right),
\end{eqnarray}
with $x_0=S_1(0)+E_1(0)+I_1(0)$ and $y_0=S_2(0)+E_2(0)+I_2(0)$.
Before we continue with the analysis of the epidemic final size, let us define the following notation for partial ordering of vectors in $\R^2$: given $X=(X_1,X_2),Y=(Y_1,Y_2)\in \R^2$, se say that
\[\begin{matrix*}[l]
X\leq Y \Longleftrightarrow X_i \leq Y_i &\mbox{ for all }&i=1,2\,,\\
X< Y \Longleftrightarrow X\leq Y \mbox{ and } X_i < Y_i &\mbox{ for some } &i=1,2\,, \\
X\ll Y\Longleftrightarrow X_i < Y_i &\mbox{ for all } &i=1,2\,.
\end{matrix*}
\]

\begin{theorem}\label{Thm.mapT}
	Let $T:\R^2\to\R^2$ be the map defined in \eqref{def.T} with $X_0=S_0+I_0+E_0$. Then, under the hypotheses \ref{H1-B} and \ref{H2-MN}, we have

	\begin{enumerate}[a)]
		\item\label{Thm.i} $T$ is componentwise increasing
		\item\label{Thm.ii} $T(S_0)=S_0 \Leftrightarrow E_0=0=I_0$
		\item\label{Thm.iii}  $E_0+I_0>0 \Leftrightarrow 0\ll T(0) \ll T(S_0)\ll S_0$
		\item\label{Thm.iv} The derivative of $T$ is componentwise increasing and given by
		\begin{equation}\label{Eq.DT1}
		DT(x,y) = \begin{pmatrix} \frac{\beta_{11}}{\gamma_1+\mu_1}T_1(x,y) & \frac{\beta_{12}}{\gamma_2+\mu_2}T_1(x,y)\\ \frac{\beta_{21}}{\gamma_1+\mu_1}T_2(x,y) & \frac{\beta_{22}}{\gamma_2+\mu_2}T_2(x,y)
		\end{pmatrix}.
		\end{equation}
	\end{enumerate}
	
\end{theorem}

\begin{proof}

The point \ref{Thm.i}) follows from hypotheses \ref{H1-B} and \ref{H2-MN}.

For point \ref{Thm.ii}), note that $T(S_0)=S_0$ is equivalent to
\[\frac{\beta_{i1}}{\gamma_1+\mu_1}(E_1(0)+I_1(0))+\frac{\beta_{i2}}{\gamma_2+\mu_2}(E_2(0)+I_2(0))=0\,\qquad i=1,2. \]
Again by hypotheses \ref{H1-B} and \ref{H2-MN}, this is true only when $E(0)=0=I(0)$, since these functions are nonnegative.

Point \ref{Thm.iii}) follows from points \ref{Thm.i}) and \ref{Thm.ii}), and point \ref{Thm.iv}) follows from the definition of $T$.

\end{proof}

It follows from the {points} \eqref{Thm.i} and \eqref{Thm.iii} in the theorem above that, for every $n\in \mathbb{N}$
\[    0\ll T(0) \ll \cdots \ll T^n(0)\ll \cdots \ll T^n(S_0) \ll \cdots \ll T(S_0)\ll S_0.  \]

Therefore, there exist the limits
\[ S^- = \lim_{n\to\infty} T^n(0) ,\;\;\; S^+ = \lim_{n\to\infty} T^n(S_0), \]
{they satisfy $S^-\leq S^+$  and are fixed points of $T$ on $[0,S_1(0)]\times[0,S_2(0)] $, since $T$ is continuous. Note that, due to the monotonicity of $T$, there is no other fixed points of $T$ in $[0,S_1(0)]\times[0,S_2(0)] \backslash [S_1^-,S_1^+]\times[S_2^-,S_2^+]$. Thus, the limit value $S^\infty$ must satisfy $S^-\leq S^\infty\leq S^+$.}

\begin{theorem}[Final size]\label{Theo.FinalSize}
	Let us assume that \eqref{H1-B} and \eqref{H2-MN} hold. If $S_0\gg 0$ and $E_0+I_0>0$, then the final {state} of the epidemic model is
	\begin{enumerate}[i)]
		\item \label{item.i} $0\neq S^\infty = \lim_{n\to \infty} T^n(0)$
		\item\label{item.ii} $I^\infty =0 = E^\infty$
		\item\label{item.iii} $R^\infty + D^\infty = N-S^\infty$
		\item\label{item.iv} $R^\infty =\Gamma (\Gamma+\M)^{-1}( N-S^\infty-R_0-D_0)+R_0$
		\item\label{item.v} $D^\infty =\M (\Gamma+\M)^{-1}( N-S^\infty-R_0-D_0)+D_0$
	\end{enumerate}
\end{theorem}

\begin{proof}
	{For item \ref{item.i}), let us show that under Hypotheses \ref{H1-B} and \ref{H2-MN}, we have $S^-=S^+$, i.e., the map $T$ has only one fixed point.}
	
	Let us suppose that $S^-<S^+$. Then, by the Taylor's Theorem we can write
	\begin{eqnarray*}
		S^+-S^- = T(S^+)-T(S^-) = \int_0^1 DT(S^-+\tau(S^+-S^-))(S^+-S^-)\;d\tau .
	\end{eqnarray*}
	From \eqref{Eq.DT1}, { we see that $DT$ is componentwise increasing by} the monotonicity of $T$. It follows that $DT(S^-+\tau(S^+-S^-))\leq DT(S^+)$ componentwise for all $\tau\in [0,1]$. Therefore, we obtain
	\begin{eqnarray}\label{Eq.S+-S-}
	S^+-S^- \leq DT(S^+)(S^+-S^-).
	\end{eqnarray}
	Since $DT(S^+)= \diag(S^+)\B(\Gamma+\M)^{-1}$, we have from Hypothesis \ref{H1-B} and \ref{H2-MN} that $DT(S^+)$ is nonnegative irreducible. Hence, by the Perron-Frobenius Theorem, let $\lambda >0$ be the dominant eigenvalue and $W\gg 0$ an associated left eigenvector. From \eqref{Eq.S+-S-} we have
	\begin{eqnarray}\label{Eq.S+-S-2}
	W^T(S^+-S^-) \leq W^T DT(S^+)(S^+-S^-) = \lambda W^T (S^+-S^-),
	\end{eqnarray}
	and thus, $\lambda \maig 1$.
	By the definition of $S^+$, we have $S_0\maig S^+$ and then $DT(S^++\tau(S_0-S^+))\maig DT(S^+)$ componentwise for all $\tau\maig 0$. Therefore,
	\begin{eqnarray*}
		T(S_0)-S^+ &=& T(S_0)-T(S^+) = \int_0^1 DT(S^++\tau(S_0-S^+))(S_0-S^+)\;d\tau \\
		&\maig& DT(S^+)(S_0-S^+).
	\end{eqnarray*}
	Multiplying the inequality above by  $W^T$,
	\begin{eqnarray*}
		W^T (T(S_0)-S^+) &\maig& W^T DT(S^+))(S_0-S^+) = \lambda W^T (S_0-S^+) \\&\maig& W^T (S_0-S^+) 
	\end{eqnarray*}
	and, since $W\gg0$ this implies that $T(S_0)\maig S_0	$, which cannot happen if $E_0+I_0>0$ by item \eqref{Thm.iii} of Theorem \ref{Thm.mapT}.
	
	The item \ref{item.ii}) was already proved.
	
	Item \ref{item.iii}) follows from item \ref{item.ii}) and \eqref{Eq.NumIndividuals}.
	
	For the item \ref{item.iv}), we can use the equations for $R$ and $D$ in \eqref{Eq.SEIRD} and obtain
	\[	R'+D'=(\Gamma+\M)I = (\Gamma+\M)\Gamma^{-1}\Gamma I = (\Gamma+\M)\Gamma^{-1}R',\]
	which implies that
	\[R(t)+D(t)-R_0-D_0=(\Gamma+\M)\Gamma^{-1}(R(t)-R_0), \qquad \forall \; t>0.\]
	Therefore, as $t\to+\infty$ we conclude that
	\[R^\infty+D^\infty-R_0-D_0=(\Gamma+\M)\Gamma^{-1}(R^\infty-R_0).\]
	Using item \ref{item.iii}), we can write
	\[N-S^\infty-R_0-D_0=(\Gamma+\M)\Gamma^{-1}(R^\infty-R_0)\]
	and thus
	\[R^\infty=\Gamma(\Gamma+\M)^{-1}(N-S^\infty-R_0-D_0)+R_0.\]
	Finally, item \ref{item.v}) follows from items \ref{item.iii}) and \ref{item.iv}).
\end{proof}

\begin{remark}
	It follows from items \ref{Thm.iii}) and \ref{Thm.iv}) that the final size of recovered and deceased classes are $R^\infty_i=\frac{\gamma_i}{\gamma_i+\mu_i}(N_i-S^\infty_i-R_i(0)-D_i(0))+R_i(0)$ and $D^\infty_i=\frac{\mu_i}{\gamma_i+\mu_i}(N_i-S^\infty_i-R_i(0)-D_i(0))+D_i(0)$ for each group $i$.
\end{remark}

\begin{remark}
	It follows from item \ref{item.i}) that the final size $S^\infty$ can also be estimated by using {only the map $T$, instead of solving numerically the system \eqref{Eq.SEIRD}. }
\end{remark}


\section{Distance estimates for the second group}\label{Section.Estimates}

 {In this section we study the following problem: given two sets of parameters, $\widetilde{\A}$, $\widetilde{\B}$, $\widetilde{\Gamma}$, $\widetilde{\M}$ and $\A, \B $, $\Gamma$, $\M$, for the system \eqref{Eq.SEIRD}, let $\widetilde{S}, \widetilde{E}, \widetilde{I}, \widetilde{R}, \widetilde{D}$ and $S, E, I, R, D$ be their respective set of solutions, we want to estimate the distance related to the solutions to second group $\norm{(\widetilde{S}_2,\widetilde{E}_2,\widetilde{I}_2,\widetilde{R}_2)-(S_2,E_2,I_2,R_2)}$ based on the the distance $\norm{(\widetilde{S}_1,\widetilde{E}_1,\widetilde{I}_1,\widetilde{R}_1,\widetilde{D}_1,\widetilde{D}_2)-(S_1,E_1,I_1,R_1,D_1,D_2)}$, i.e.,
based on the distances between solutions of the first group plus on the distance between the evolution of deaths in group 2 too.}

We start with the following lemma on the ratio of susceptible class:

\begin{lemma}\label{S.ratio}The solutions $S_1(t)$ and $S_2(t)$ of susceptible individuals satisfy
\[\displaystyle \frac{S_1(t)}{S_2(t)}= \frac{S_1(0)}{S_2(0)}\exp\left\{ \int_0^t\left(\left(\beta_{21}-\beta_{11}\right)I_{1}(s)+\left(\beta_{22}-\beta_{12}\right)I_{2}(s)\right)\,ds \right\}.\]
\end{lemma}
\begin{proof} From the equations for $S_1$ and $S_2$ in \eqref{Eq.SEIRDcomponentwise}, we have
	\begin{eqnarray*}
		\displaystyle\left[\frac{S_1}{S_2}\right]'&=&\frac{(\beta_{21}I_1+\beta_{22}I_2)S_1S_2-(\beta_{11}I_1+\beta_{12}I_2)S_1S_2}{S_2^2}\\
		&=&\big(\left(\beta_{21}-\beta_{11}\right)I_{1}+\left(\beta_{22}-\beta_{12}\right)I_{2}\big)\frac{S_1}{S_2}.
	\end{eqnarray*}
	and thus, the result follows.
\end{proof}

The following two elementary results will be useful for the estimates in Theorem \ref{Theo.Estim}, therefore we include them with proofs.

\begin{remark}\label{Remark.f(e)} Let $a,b>0$ and $f:[0,+\infty)\to \R$ defined by
	\[ f(\varepsilon) = (1-b\varepsilon)e^{-a\varepsilon} .\]
	Then, {the Mean Value Theorem implies that, for every $\varepsilon >0$, there holds}
	\[\left|1-(1-b\varepsilon)e^{-a\varepsilon} \right| \leq \max_{\xi\in[0,\varepsilon]}|f'(\xi)| \varepsilon= |f'(0)|\varepsilon = (a+b) \varepsilon .        \]
\end{remark}

\begin{lemma}\label{Remark.g(s)} 
	Let $a,b>0$ and $g:[0,T]\to \R$ defined by
	\[ g(s) = ae^{a(s-T)}-be^{b(s-T)}. \]
	Then there exists at most one $\tau\in[0,T)$ such that $g(\tau)=0$. Also, if  {$h:[0,T]\to\R_+$} is nonincreasing, then
	\begin{equation}
	\int_0^Tg(s)h(s)\,ds \leq \left\{  \begin{matrix}
	h(0)(e^{-bT}-e^{-aT}),\hfill & \mbox{if } 0<b<a \\ 
	h(0)(e^{-bT}-e^{-aT})+(h(0)-h(T))\ln\left(\frac{b}{a}\right), & \mbox{if } 0<a<b 
	\end{matrix} \right.  .
	\end{equation}
\end{lemma}

\begin{proof}
	To see this, note that $g(\tau)=0$ if and only if $\tau=T-\frac{1}{a-b}\ln\left(\frac{a}{b}\right)\in[0,T)$, so it changes sign only once. 
	If $0<b<a$, we have $g(T)=a-b>0$ and then $g$ nonnegative on $[0,T]$ or it changes sign at $s=\tau$. In this case, for any  $h$ nonincreasing nonnegative function we can write 
	\[
	\int_0^T\!g(s)h(s)ds = \int_0^\tau \!g(s)h(s)\,ds+\int_{\tau}^T\!g(s)h(s)\,ds \leq h(\tau)\int_0^\tau\! g(s)\,ds+h(\tau)\int_{\tau}^T\!g(s)\,ds
	\]
	and the estimate follows. In the case where $g$ is nonnegative, the estimate is immediate.
	
	If $0<a<b$, we have $g(T)=a-b<0$ and then
	\begin{eqnarray*}
		\int_0^Tg(s)h(s)\,ds &=& \int_0^\tau g(s)h(s)\,ds+\int_{\tau}^Tg(s)h(s)\,ds \\
		&\leq& h(0)\left(e^{a(\tau-T)}-e^{-aT}-e^{b(\tau-T)}+e^{-bT}\right)\!+\!h(T)\left(e^{b(\tau-T)}\!-\!e^{a(\tau-T)}\right)\\
		&=& h(0)\left(e^{-bT}-e^{-aT}\right)+(h(0)-h(T))\left(e^{a(\tau-T)}-e^{b(\tau-T)}\right).
	\end{eqnarray*}
	By the Mean Value Theorem, there exists a $a<\xi<b$ satisfying
	\begin{eqnarray*}
		e^{a(\tau-T)}-e^{b(\tau-T)}&=&(T-\tau)(b-a)e^{\xi(\tau-T)}\leq (T-\tau)(b-a)e^{a(\tau-T)}\\
		&=&\ln\left(\frac{b}{a}\right)\left(\frac{a}{b}\right)^{\frac{a}{b-a}} \leq \ln\left(\frac{b}{a}\right),
	\end{eqnarray*}
	where for the last equality we used {that $T-\tau=\frac{1}{a-b}\ln\left(\frac{a}{b}\right)$.}
\end{proof}

For the main result of this section, we will denote the norm $L^\infty(0,T)$ by $\norm{\,.\,}$.

\begin{theorem}\label{Theo.Estim} 
	Let $(\widetilde{S},\widetilde{E},\widetilde{I},\widetilde{R},\widetilde{D})$ and $(S,E,I,R,D)$ be two solutions for (\ref{Eq.SEIRD}) corresponding to two sets of parameters $\widetilde{\A}$, $\widetilde{\B}$, $\widetilde{\Gamma}$, $\widetilde{\M}$ and $\A, \B $, $\Gamma$, $\M$ respectively, with the same initial conditions. If 
	\begin{equation}\label{Eq.group1}
		\norm{(\widetilde{S}_1,\widetilde{E}_1,\widetilde{I}_1,\widetilde{R}_1,\widetilde{D}_1)-(S_1,E_1,I_1,R_1,D_1)}\leq\varepsilon
	\end{equation}
	and
	\begin{multicols}{2}
		\begin{enumerate}[i)]
			\item\label{Hyp.i}$\frac{\widetilde{\beta}_{2i}-\widetilde{\beta}_{1i}}{\widetilde{\mu}_i}=\frac{\beta_{2i}-\beta_{1i}}{\mu_i}, \mbox{ for }i=1,2;$
			\item\label{Hyp.ii}$\frac{\widetilde{\alpha}_1}{\widetilde{\alpha}_2}=\frac{\alpha_1}{\alpha_2};$
			\item\label{Hyp.iii}$\widetilde{\gamma}_1-\widetilde{\gamma}_2=\gamma_1-\gamma_2;$
			\item\label{Hyp.iv}$\widetilde{\mu}_1-\widetilde{\mu}_2=\mu_1-\mu_2;$
			\item\label{Hyp.v} $\norm{D_2-\widetilde{D}_2}<\varepsilon$
		\end{enumerate}
	\end{multicols}
	then 
	\[\norm{(\widetilde{S}_2,\widetilde{E}_2,\widetilde{I}_2,\widetilde{R}_2)-(S_2,E_2,I_2,R_2)}<K(\varepsilon)\varepsilon,\]
	where $K(\varepsilon)$ is bounded as $\varepsilon \to 0$.  Moreover, if $\widetilde{\beta}_{22}-\widetilde{\beta}_{12}=\beta_{22}-\beta_{12}$, then the {hypothesis} \eqref{Hyp.v} is not necessary.
\end{theorem}

\begin{proof} For simplicity let us define $V_i=\frac{\widetilde{\beta}_{2i}-\widetilde{\beta}_{1i}}{\widetilde{\mu}_i}=\frac{\beta_{2i}-\beta_{1i}}{\mu_i}, \mbox{ for }i=1,2$. 
	
	\noindent {\bf Step 1:} For the susceptible part, by Lemma \ref{S.ratio} and the equation for $D_1$ and $D_2$ we have
	\[\displaystyle \frac{\widetilde{S}_2}{S_2}=\displaystyle\frac{\widetilde{S}_1}{S_1} \exp\{ V_1(D_{1}(t)-\widetilde{D}_1(t))+V_2(D_2(t)-\widetilde{D}_2(t))\}.\]
	{Firstly, let us assume that for a given $t>0$ we have $\widetilde{S}_2(t)\leq S_2(t)$. Since $|S_1(t)-\widetilde{S}_1(t)|<\varepsilon$ by hypothesis, we obtain}
	\begin{eqnarray}
	0 &\leq& S_2(t)-\widetilde{S}_2(t)\leq {S}_2(t)\left(1-\frac{\widetilde{S}_2(t)}{S_2(t)}\right) \nonumber  \\&=& {S}_2(t)\left(1-\frac{\widetilde{S}_1(t)}{S_1(t)}e^{V_1(D_{1}(t)-\widetilde{D}_1(t))+V_2(D_2(t)-\widetilde{D}_2(t))}\right)\nonumber \\
	&\leq& {S}_2(t)\left(1-\left(1-\frac{\varepsilon}{S_1(t)}\right)e^{V_1(D_{1}(t)-\widetilde{D}_1(t))+V_2(D_2(t)-\widetilde{D}_2(t))}\right) .\label{eq.Estimate.S.with.D}
	\end{eqnarray}
	Since $|D_{i}(t)-\widetilde{D}_i(t)|\leq \varepsilon$, we can write
	\begin{eqnarray}
		0 &\leq& S_2(t)-\widetilde{S}_2(t)\leq {S}_2(t)\left(1-\left(1-\frac{\varepsilon}{S_1(t)}\right)e^{-(|V_1|+|V_2|)\varepsilon}\right) \nonumber\\
		&\leq& S_2(t)\left(|V_1|+|V_2|+\frac{1}{S_1(t)}\right)\varepsilon ,\label{Eq.St2<S2}
	\end{eqnarray}
	where, for the last inequality, we used the Remark \ref{Remark.f(e)} for $a=|V_1|+|V_2|$ and $b=1/S_1(t)$.
	
	{In the case where $S_2(t)\leq\widetilde{S}_2(t)$ for a given $t>0$}, we proceed analogously and obtain
	\begin{eqnarray}
		0 \leq \widetilde{S}_2(t)-S_2(t)&\leq& \widetilde{S}_2(t)\left(1-\frac{{S}_1(t)}{\widetilde{S}_1(t)}e^{-V_1(D_{1}(t)-\widetilde{D}_1(t))-V_2(D_2(t)-\widetilde{D}_2(t))}\right)\nonumber\\
		&\leq& \widetilde{S}_2(t)\left(1-\frac{{S}_1(t)}{{S}_1(t)+\varepsilon}e^{-(|V_1|+|V_2|)\varepsilon}\right)\nonumber\\
		&\leq &S_2(0)\left(1-\left(1-\frac{\varepsilon}{{S}_1(t)+\varepsilon}\right)e^{-(|V_1|+|V_2|)\varepsilon}\right),\label{Eq.S2<St2}
	\end{eqnarray}
	since the initial condition is the same for $S_2$ and $\widetilde{S}_2$. {Thus, for every fixed $T>0$, we can apply \eqref{Eq.S2<St2} and \eqref{Eq.St2<S2} to each $t\in[0,T]$ and obtain that}
	\begin{equation}
	\norm{\widetilde{S}_2-S_2}\leq K_S \varepsilon ,
	\end{equation}
	where $K_S=S_2(0)\left(|V_1|+|V_2|+\frac{1}{S_1(T)}\right).$

	\noindent {\bf Step 2:} Suppose that $E_2(t)\geq\widetilde{E}_2(t)$. Note that 
	\[S_2(t) = S_2(t)e^{-\alpha_2 t}+\alpha_2\int_0^t S(t)e^{\alpha_2(s-t)}\]  
	and then, by Lemma \ref{S.ratio} and Lemma \ref{explicitLemma} we have
	\begin{eqnarray*}
		E_2(t)-\widetilde{E}_2(t)=& &-S_2(t)+\alpha_2\int_0^te^{\alpha_2(s-t)}S_2(s)ds + (S_2(0)+E_2(0))e^{-\alpha_2t}\\
		& &+\widetilde{S}_2(t)- \widetilde{\alpha}_2\int_0^te^{\widetilde{\alpha}_2(s-t)}\widetilde{S}_2(s)ds  - (S_2(0)+E_2(0))e^{-\widetilde{\alpha}_2t}  \nonumber \\ 
		= & & \alpha_2\int_0^te^{\alpha_2(s-t)}(S_2(s)-S_2(t))ds + (S_2(0)+E_2(0)-S_2(t))e^{-\alpha_2t}\\
		& &- \widetilde{\alpha}_2\int_0^te^{\widetilde{\alpha}_2(s-t)}(\widetilde{S}_2(s)-\widetilde{S}_2(t))ds  - (S_2(0)+E_2(0)-\widetilde{S}_2(t))e^{-\widetilde{\alpha}_2t}\\
		\leq& &  \alpha_2\int_0^te^{\alpha_2(s-t)}(S_2(s)-S_2(t))ds + (S_2(0)+E_2(0)-S_2(t))e^{-\alpha_2t}\\
		& &- \widetilde{\alpha}_2\int_0^te^{\widetilde{\alpha}_2(s-t)}(S_2(s)-S_2(t)-2K_S\varepsilon)ds  \\& &- (S_2(0)+E_2(0)-S_2(t)-K_S\varepsilon)e^{-\widetilde{\alpha}_2t}\\
		\leq & & \int_0^t \left(   \alpha_2e^{\alpha_2(s-t)} - \widetilde{\alpha}_2e^{\widetilde{\alpha}_2(s-t)}  \right)(S_2(s)-S_2(t))ds+ K_S\varepsilon e^{-\widetilde{\alpha}_2t}\\ & &+ (S_2(0)+E_2(0)-S_2(t))\left(e^{-\alpha_2t}-e^{-\widetilde{\alpha}_2t}\right)
		+ 2\widetilde{\alpha}_2K_S\varepsilon(1-e^{-\widetilde{\alpha}_2t}) .
	\end{eqnarray*}
	Now, supposing that $\widetilde{\alpha}_2<\alpha_2$, by the Lemma \ref{Remark.g(s)} for $a=\alpha_2$, $b=\widetilde{\alpha}_2$ and $h(s)=S_2(s)-S_2(t)$, we can estimate
	\begin{eqnarray*}
		E_2(t)\!-\!\widetilde{E}_2(t) &\!\leq\!&  (S_2(0)\!-\!S_2(t))\!\left( e^{-\widetilde{\alpha}_2t} \!-\!e^{-\alpha_2t} \right) \!+\! (S_2(0)\!+\!E_2(0)\!-\!S_2(t))\!\left(e^{-\alpha_2t}\!-\!e^{-\widetilde{\alpha}_2t}\right)\\
		& &\;\;\;\;+ 2\widetilde{\alpha}_2K_S\varepsilon +(1-2\widetilde{\alpha}_2)K_S\varepsilon e^{-\widetilde{\alpha}_2t}\\
		&=& -E_2(0)\left( e^{-\widetilde{\alpha}_2t} -e^{-\alpha_2t} \right) + \left(2\widetilde{\alpha}_2 +(1-2\widetilde{\alpha}_2) e^{-\widetilde{\alpha}_2t}\right)K_S\varepsilon \\
		&\leq& -E_2(0)(\alpha_2 - \widetilde{\alpha}_2)te^{-\alpha_2t}  +K_S\varepsilon 
	\end{eqnarray*}
	and, if $\widetilde{\alpha}_2>\alpha_2$ the Lemma \ref{Remark.g(s)} implies that
	\begin{eqnarray*}
		E_2(t)-\widetilde{E}_2(t) &\leq&  (S_2(0)-S_2(t))\left( e^{-\widetilde{\alpha}_2t} -e^{-\alpha_2t} \right) + (S_2(0)-S_2(t))\ln\left(\frac{\widetilde{\alpha}_2}{{\alpha}_2}\right)\\
		& &+(S_2(0)+E_2(0)-S_2(t))\left(e^{-\alpha_2t}-e^{-\widetilde{\alpha}_2t}\right)+K_S\varepsilon\\
		&=& -E_2(0)\left( e^{-\widetilde{\alpha}_2t} -e^{-\alpha_2t} \right) + (S_2(0)-S_2(t))\ln\left(\frac{\widetilde{\alpha}_2}{{\alpha}_2}\right)+K_S\varepsilon\\
		&=& E_2(0)(\widetilde{\alpha}_2-\alpha_2)te^{-\alpha_2t}  + (S_2(0)-S_2(t))\ln\left(\frac{\widetilde{\alpha}_2}{{\alpha}_2}\right) +K_S\varepsilon .
	\end{eqnarray*}

	It remains to show that the parameters $\widetilde{\alpha}_2$ and ${\alpha}_2$ are near each other under the {hypothesis} of the Theorem. To see this, we can integrate \eqref{Eq.EeS} for $i=1$ and obtain
	\begin{equation}\label{Eq.IntegralE}
	\alpha_1 \int_0^t E_1(s)\, ds= S_1(0)-S_1(t)+E_1(0)-E_1(t) .
	\end{equation}
	Thus, using the same argument for the parameter $\widetilde{\alpha}_1$ and the hypothesis \eqref{Hyp.ii} and supposing {without loss of generality} that $\widetilde{\alpha}_2>\alpha_2$, we have for every $t>0$ that
	\begin{eqnarray*}
		1\leq \frac{\widetilde{\alpha}_2}{\alpha_2}&=&\frac{\widetilde{\alpha}_1}{\alpha_1} =\left(\frac{S_1(0)-\widetilde{S}_1(t)+E_1(0)-\widetilde{E}_1(t)}{S_1(0)-S_1(t)+E_1(0)-E_1(t)}\right)\left(\frac{\int _0^tE_1 ds}{\int_0^t \widetilde{E}_1 ds}\right)\\
		&\leq&\left(1+\frac{2\varepsilon}{S_1(0)-S_1(t)+E_1(0)-E_1(t)}\right)\left(1+\frac{\varepsilon t}{\int_0^t \widetilde{E}_1 ds}\right) .
	\end{eqnarray*}
	Therefore, using for example $t=1$, there exists constants $C_1,C_2>0$ such that
	\begin{eqnarray*}
		1 &\leq \frac{\widetilde{\alpha}_2}{\alpha_2}&\leq 1+C_1\varepsilon +C_2\varepsilon^2 \\
		\mbox{and thus} \hspace{0.23\textwidth} 0&<\widetilde{\alpha}_2-\alpha_2&\leq \alpha_2C_1\varepsilon +\alpha_2C_2\varepsilon^2 \hspace{0.35\textwidth} .
	\end{eqnarray*}

	Therefore, we conclude that, for $\widetilde{\alpha}_2<\alpha_2$, there holds
\[
	E_2(t)-\widetilde{E}_2(t) \leq -E_2(0)(\alpha_2 - \widetilde{\alpha}_2)te^{-\alpha_2t} + K_S\varepsilon 
\]
	and, for $\widetilde{\alpha}_2>\alpha_2$ we have
\[
	E_2(t)-\widetilde{E}_2(t) \leq \left(\alpha_2E_2(0)te^{-\alpha_2t}  + S_2(0)-S_2(t)\right) (C_1+C_2\varepsilon )\varepsilon+K_S\varepsilon .
\]
	
	Therefore, 
	\begin{equation}\label{Eq.EstimateE}
	\norm{E_2-\widetilde{E}_2}\leq K_E(\varepsilon)\varepsilon ,
	\end{equation}
	where
	\[K_E(\varepsilon) = \left(\frac{E_2(0)}{e}+S_2(0)-S_2(T)\right)(C_1+C_2\varepsilon) +K_S ,\]
	since for every $\lambda>0$, the function $g(t):= te^{-\lambda t}$ satisfies $g(t)\leq \frac{1}{\lambda e}$.
	
	\noindent {\bf Step 3:} For the functions $I_2$ and $\widetilde{I}_2(t)$, we can use \eqref{Eq.explicitI} and write
	\begin{eqnarray}
	I_2(t)-\widetilde{I}_2(t) &=& I_2(0)\left(e^{-(\gamma_2+\mu_2) t}-e^{-(\widetilde{\gamma}_2+\widetilde{\mu}_2) t}\right) \nonumber \\
	& & +\int_0^t \left(\alpha_2 e^{(\gamma_2+\mu_2)(s-t)}E_2(s)-\alpha_2 e^{(\widetilde{\gamma}_2+\widetilde{\mu}_2)(s-t)}\widetilde{E}_2(s)\right)\,ds \nonumber \\
	&=:& J_1(t)+J_2(t) , \label{Eq.EstimateI.J1J2}
	\end{eqnarray}
	where $J_1(t)$ and $J_2(t)$ are defined by the first and second term on the right-hand side of the first equality. Firstly we can use the Mean Value Theorem to obtain
	\begin{eqnarray*}
		|J_1(t)| = |I_2(0)|\left|e^{-(\gamma_2+\mu_2) t}-e^{-(\widetilde{\gamma}_2+\widetilde{\mu}_2) t}\right| \leq |I_2(0)||\gamma_2+\mu_2 -\widetilde{\gamma}_2-\widetilde{\mu}_2|te^{-\min\{\gamma_2+\mu_2,\widetilde{\gamma}_2+\widetilde{\mu}_2\}t}.
	\end{eqnarray*}
	In order to prove that the {first} term on the right-hand side of the inequality is bounded by $\varepsilon$, note that the equations for $R_1(t)$ and $D_1(t)$ imply
	\[R'_1(t)+D'_1(t)=(\gamma_1+\mu_1)I_1(t)\]
	and therefore, supposing w.l.g. $\gamma_1+\mu_1<\widetilde{\gamma}_1+\widetilde{\mu}_1$,
	\begin{eqnarray*}
		1< \frac{\widetilde{\gamma}_1+\widetilde{\mu}_1}{\gamma_1+\mu_1} &=&\frac{\widetilde{R}_1(t)+\widetilde{D}_1(t)-R_1(0)-D_1(0)}{R_1(t)+D_1(t)-R_1(0)-D_1(0)}\;\frac{\int_0^t I_1(s)\,ds}{\int_0^t \widetilde{I}_1(s)\,ds}\\
		&\leq&\left(1+\frac{2\varepsilon}{R_1(t)+D_1(t)-R_1(0)-D_1(0)}\right)\left(1+\frac{\varepsilon t}{\int_0^t \widetilde{I}_1(s)\,ds} \right)\\
		&=&\left(1+\frac{2\varepsilon}{R_1(t)+D_1(t)-R_1(0)-D_1(0)}\right)\left(1+\frac{\varepsilon \widetilde{\mu}_1t}{\widetilde{D}_1(t)-\widetilde{D}_1(0)} \right) ,
	\end{eqnarray*}
	where the last equality comes from the equation for $\widetilde{D}_1(t)$.
	
	Since the inequality above holds for every $t>0$, there exist constants $C_3$ and $C_4$ such that
	\begin{eqnarray*}
		\frac{\widetilde{\gamma}_1+\widetilde{\mu}_1}{\gamma_1+\mu_1} &\leq& 1+C_3\varepsilon +C_4\varepsilon^2
	\end{eqnarray*}
	and then
	\begin{eqnarray*}
		0\leq \widetilde{\gamma}_1+\widetilde{\mu}_1-\gamma_1-\mu_1 &\leq& (\gamma_1+\mu_1)\left(C_3\varepsilon +C_4\varepsilon^2\right) .
	\end{eqnarray*}
	Since by hypotheses \eqref{Hyp.iii} and \eqref{Hyp.iv} we have $\widetilde{\gamma}_2+\widetilde{\mu}_2-\gamma_2-\mu_2=\widetilde{\gamma}_1+\widetilde{\mu}_1-\gamma_1-\mu_1$, we conclude that
	\begin{eqnarray}\label{Eq.EstimateJ1}
	|J_1(t)| \leq |I_2(0)|te^{-\min\{\gamma_2+\mu_2,\widetilde{\gamma}_2+\widetilde{\mu}_2\}t}(\gamma_1+\mu_1)\left(C_3\varepsilon +C_4\varepsilon^2\right).
	\end{eqnarray}

	Now for the term $J_2(t)$, we can decompose it as
	\begin{eqnarray}
	J_2(t) = H_1(t)+H_2(t)+H_3(t) \label{Eq.J2-H1H2H3}
	\end{eqnarray}
	where 
	\begin{eqnarray*}
		H_1(t)&:=&(\alpha_2-\widetilde{\alpha}_2)\int_0^t e^{(\gamma_2+\mu_2)(s-t)}E_2(s)\,ds ,\\
		H_2(t)&:=&\widetilde{\alpha}_2\int_0^t \left(e^{(\gamma_2+\mu_2)(s-t)}-e^{(\widetilde{\gamma}_2+\widetilde{\mu}_2)(s-t)}\right)E_2(s)\,ds ,\\
		H_3(t)&:=&\widetilde{\alpha}_2\int_0^t e^{(\widetilde{\gamma}_2+\widetilde{\mu}_2)(s-t)}\left(E_2(s)-\widetilde{E}_2(s)\right)\,ds .
	\end{eqnarray*}
	For the $H_1(t)$, we can use \eqref{Eq.IntegralE} and obtain
	\begin{eqnarray}
	\left|H_1(t)\right|&\leq& |\alpha_2-\widetilde{\alpha}_2|\frac{S_2(0)-S_2(t)+E_2(0)-E_2(t)}{\alpha_2} \nonumber \\
	&\leq&\left(S_2(0)-S_2(t)+E_2(0)-E_2(t)\right)(C_1+C_2\varepsilon)\varepsilon . \label{Eq.EstimateH1}
	\end{eqnarray}
	For the $H_2(t)$, we can use a similar argument as the one used for $J_1(t)$ and $H_1(t)$:
	\begin{eqnarray}
	\left|H_2(t)\right|&\leq& \widetilde{\alpha}_2|\gamma_2+\mu_2-\widetilde{\gamma}_2-\widetilde{\mu}_2| \int_0^t(t-s)e^{\min\{\gamma_2+\mu_2,\widetilde{\gamma}_2+\widetilde{\mu}_2\}(s-t)}E_2(s)\,ds \nonumber\\
	&\leq& \widetilde{\alpha}_2\frac{|\gamma_2+\mu_2-\widetilde{\gamma}_2-\widetilde{\mu}_2|}{\min\{\gamma_2+\mu_2,\widetilde{\gamma}_2+\widetilde{\mu}_2\}e} \int_0^tE_2(s)\,ds \nonumber\\
	&\leq& \frac{\widetilde{\alpha}_2}{\alpha_2} \frac{|\gamma_2+\mu_2-\widetilde{\gamma}_2-\widetilde{\mu}_2|}{\min\{\gamma_2+\mu_2,\widetilde{\gamma}_2+\widetilde{\mu}_2\}e}\left(S_2(0)-S_2(t)+E_2(0)-E_2(t)\right) \nonumber \\
	&\leq & \frac{\widetilde{\alpha}_2}{\alpha_2}\,\frac{S_2(0)-S_2(t)+E_2(0)-E_2(t)}{\min\{\gamma_2+\mu_2,\widetilde{\gamma}_2+\widetilde{\mu}_2\}e} (\gamma_1+\mu_1)\left(C_3\varepsilon +C_4\varepsilon^2\right). \label{Eq.EstimateH2}
	\end{eqnarray}

	Finally, we use  \eqref{Eq.EstimateE} to estimate $H_3(t)$ as
	\begin{eqnarray}\label{Eq.EstimateH3}
	\left|H_3(t)\right|\leq \frac{\widetilde{\alpha}_2}{\widetilde{\gamma}_2-\widetilde{\mu}_2}(1-e^{-(\widetilde{\gamma}_2+\widetilde{\mu}_2)t})K_E \varepsilon
	\end{eqnarray}
	
	and, from \eqref{Eq.EstimateI.J1J2}, \eqref{Eq.EstimateJ1}, \eqref{Eq.J2-H1H2H3}, \eqref{Eq.EstimateH1}, \eqref{Eq.EstimateH2} and \eqref{Eq.EstimateH3}, we obtain
	\begin{eqnarray*}
		\left|I_2(t)-\widetilde{I}_2(t)\right|&\leq& |I_2(0)|te^{-\min\{\gamma_2+\mu_2,\widetilde{\gamma}_2+\widetilde{\mu}_2\}t}(\gamma_1+\mu_1)\left(C_3\varepsilon +C_4\varepsilon^2\right)\\
		&&+\left(S_2(0)-S_2(t)+E_2(0)-E_2(t)\right)(C_1+C_2\varepsilon)\varepsilon\\
		&&+\frac{\widetilde{\alpha}_2}{\alpha_2}\,\frac{S_2(0)-S_2(t)+E_2(0)-E_2(t)}{\min\{\gamma_2+\mu_2,\widetilde{\gamma}_2+\widetilde{\mu}_2\}e} (\gamma_1+\mu_1)\left(C_3\varepsilon +C_4\varepsilon^2\right)\\
		&&+ \frac{\widetilde{\alpha}_2}{\widetilde{\gamma}_2-\widetilde{\mu}_2}(1-e^{-(\widetilde{\gamma}_2+ \widetilde{\mu}_2)t})K_E \varepsilon .
	\end{eqnarray*}
	Therefore, we have $\norm{I_2-\widetilde{I}}\leq K_I(\varepsilon)\varepsilon$
	where
	\begin{eqnarray*}
		K_I(\varepsilon)&=&\left(|I_2(0)|+\frac{\widetilde{\alpha}_2}{\alpha_2}\left((S_2(0)-S_2(T)+E_2(0)\right)\right)\frac{(\gamma_1+\mu_1)\left(C_3 +C_4\varepsilon\right)}{\min\{\gamma_2+\mu_2,\widetilde{\gamma}_2+\widetilde{\mu}_2\}e}\\
		&&+\left(S_2(0)-S_2(T)+E_2(0)\right)(C_1+C_2\varepsilon)
		+ \frac{\widetilde{\alpha}_2}{\widetilde{\gamma}_2-\widetilde{\mu}_2}K_E  .
	\end{eqnarray*}

	\noindent {\bf Step 4:} For the solutions $R_2$ and $\widetilde{R}_2$, we use \eqref{Eq.NumIndividuals} \[S_2+E_2+I_2+R_2+D_2=N_2=\widetilde{S}_2+\widetilde{E}_2+\widetilde{I}_2+\widetilde{R}_2+\widetilde{D}_2\] 
	and thus
	\[\norm{R_2-\widetilde{R}_2} \leq \left(K_S+K_E+K_I+1\right)\varepsilon.\]	
\end{proof}
\begin{remark} In many cases the hypothesis \ref{Hyp.ii}-\ref{Hyp.iv} can be retrieved from clinical data under the assumption that some of the parameters are equal for both populations. Hence, in general, condition \ref{Hyp.i} seems to be the most difficult to verify for practical applications. 
\end{remark}

\begin{remark}
{Note that, given a fixed model with parameters $\widetilde{\A}$, $\widetilde{\B}$, $\widetilde{\Gamma}$, $\widetilde{\M}$, the conclusion of the theorem holds only for all the solutions $S_2,E_2,I_2,R_2,D_2$ to models whose parameters $\A, \B $, $\Gamma$, $\M$ satisfy \eqref{Hyp.i}-\eqref{Hyp.iv} with respect to the fixed ones.}
\end{remark}

\begin{remark} Instead of using the equation for $D_2$ in Step 1, we could use the equation for $R_2$. Then, the estimates would depend on $\frac{\beta_{2i}-\beta_{1i}}{\gamma_i}$ instead of $\frac{\beta_{2i}-\beta_{1i}}{\mu_i}$.
	In this case, hypotheses \eqref{Hyp.i} and \eqref{Hyp.v} would change. We choose to base our estimates on death rates because data on number of deaths are more reliable for applications.

\end{remark}
\begin{remark} Alternatively, if $\B(\Gamma+\M)^{-1}=\widetilde{\B}(\widetilde{\Gamma}+\widetilde{\M})^{-1}$, i.e., if
	\[\frac{\beta_{ij}}{\gamma_j+\mu_j} = \frac{\widetilde{\beta}_{ij}}{\widetilde{\gamma}_j+\widetilde{\mu}_j} \mbox{ for } i,j=1,2 \;,   \]
	we can prove the estimates in Step 3 by using \eqref{Eq.Fconstante} and then
	\begin{eqnarray*}
		\frac{\beta_{11}}{\gamma_1+\mu_1}(S_1(t)+E_1(t)+I_1(t))+\frac{\beta_{12}}{\gamma_2+\mu_2}(S_2(t)+E_2(t)+I_2(t))&=&P_1 ,\\
		\frac{\widetilde{\beta}_{11}}{\widetilde{\gamma}_1+\widetilde{\mu}_1} (\widetilde{S}_1(t)+\widetilde{E}_1(t)+\widetilde{I}_1(t))+\frac{\widetilde{\beta}_{12}}{\widetilde{\gamma}_2+\widetilde{\mu}_2}(\widetilde{S}_2(t)+\widetilde{E}_2(t)+\widetilde{I}_2(t))&=&\widetilde{P}_1 ,
	\end{eqnarray*}
	where $P_1$ and $\widetilde{P}_1$ are constants depending on the matrices $\B(\Gamma+\M)^{-1}$ and $\widetilde{\B}(\widetilde{\Gamma}+\widetilde{\M})^{-1}$, and on the initial conditions. By hypothesis, $P_1=\widetilde{P}_1$. Then, from these two identities, we obtain
	\begin{eqnarray*}
		I_2(t)-\widetilde{I}_2(t) &=& \frac{\gamma_2+\mu_2}{\beta_{12}}\ln\left(\frac{S_1(t)}{\widetilde{S}_1(t)}\right)-S_2(t)+\widetilde{S}_2(t)-E_2(t)+\widetilde{E}_2(t)\\
		&&-\frac{\gamma_2+\mu_2}{\gamma_1+\mu_1}\frac{\beta_{11}}{\beta_{12}}\left(S_1(t)-\widetilde{S}_1(t)+E_1(t)-\widetilde{E}_1(t)+I_1(t)-\widetilde{I}_1(t)\right) ,
	\end{eqnarray*}
	and we can estimate
	\begin{eqnarray*}
		\left|I_2(t)-\widetilde{I}_2(t)\right| &=& \left(\frac{\gamma_2+\mu_2}{\beta_{12}}\frac{1}{\widetilde{S}_1(t)}+3 \frac{\gamma_2+\mu_2}{\gamma_1+\mu_1}\frac{\beta_{11}}{\beta_{12}} +K_S +K_E\right)\varepsilon.
	\end{eqnarray*}
\end{remark}

\begin{remark}\label{Remark.R}
	Although the argument in Step 4 is enough for the proof of the theorem, the estimate is not very good since it considers each of the previous classes separately. One can improve the the estimate by noticing from the $R$ and $D$ equations in \eqref{Eq.SEIRDcomponentwise} that
\[	R_2'(t)=\gamma_2 I_2(t) = \frac{\gamma_2}{\mu_2} D_2'(t)\]
	and therefore
	\[R_2(t)-R_2(0)= \frac{\gamma_2}{\mu_2} (D_2(t)-D_2(0)) \;, \;\; \forall t\geq 0.\]
	Thus, using the same argument for $\widetilde{R}_2(t)$, we obtain
	\begin{equation}\label{Eq.EstimateR}
	\left|R_2(t)-\widetilde{R}_2(t)\right|\leq\left| \frac{\gamma_2}{\mu_2} -  \frac{\widetilde{\gamma}_2}{\widetilde{\mu}_2} \right| \left| D_2(t)-D_2(0)\right| + \frac{\widetilde{\gamma}_2}{\widetilde{\mu}_2}\left|  D_2(t)- \widetilde{D}_2(t)\right| ,
	\end{equation}
	and thus using the hypotheses \eqref{Hyp.iii}, \eqref{Hyp.iv} and \eqref{Hyp.v}, we obtain the estimate.
\end{remark}

\begin{remark}\label{Remark.better}
	We could obtain better estimate for the distances between solutions for the second group by considering separate {distances} on the first group. For example, defining
	\begin{equation}\label{Eq.Epsilons}
	\begin{aligned}
	&&\varepsilon_S:=\norm{S_1-\widetilde{S}_1}\,,\;\varepsilon_E:=\norm{E_1-\widetilde{E}_1}\,,\; \varepsilon_I:=\norm{I_1-\widetilde{I}_1},\\ 
	&&\varepsilon_R:=\norm{R_1-\widetilde{R}_1} \,,\;\varepsilon_{D_1}:=\norm{D_1-\widetilde{D}_1},\;\varepsilon_{D_2}:=\norm{D_2-\widetilde{D}_2},
	\end{aligned}
	\end{equation}
	we can rewrite the estimate for the susceptible class in the second group as $\norm{S_2-\widetilde{S}_2}\leq \K_S(T)$ where
	\begin{equation}\label{Eq.EstimateS.better}
	\K_S(T):= {S}_2(T)\left(1-\left(1-\frac{\varepsilon_S}{S_1(T)}\right)e^{-\left(|V_1|\varepsilon_{D_1}+|V_2|\varepsilon_{D_2}\right)}\right) .
	\end{equation}
	
	Analogously, we have $\norm{E_2-\widetilde{E}_2}\leq \K_E(T)$  and $\norm{I_2-\widetilde{I}_2}\leq \K_I(T)$ where
	\begin{eqnarray}\label{Eq.EstimateE.better}
	\K_E(T) &=& E_2(0)|\widetilde{\alpha}_2-\alpha_2|te^{-\alpha_2t} \\&&+(S_2(0)-S_2(T))\max\left\{0,\ln\left( \frac{\widetilde{\alpha}_2}{\alpha_2}\right)\right\}+\K_S(T)\nonumber
	\end{eqnarray}
	and
	\begin{equation}\label{Eq.EstimateI.better}
	\begin{aligned}
	\K_I(T)&= \frac{|\gamma_2+\mu_2 -\widetilde{\gamma}_2-\widetilde{\mu}_2|}{\min\{\gamma_2+\mu_2, \widetilde{\gamma}_2+\widetilde{\mu}_2\} e}  \left(I_2(0)+\frac{\widetilde{\alpha}_2}{\alpha_2} \left(S_2(0)-S_2(T)+E_2(0)\right)\right)    \\ 
	&+\left|1-\frac{\widetilde{\alpha}_2}{\alpha_2}\right|\left(S_2(0)-S_2(T)+E_2(0)\right) \frac{\widetilde{\alpha}_2}{\widetilde{\gamma}_2-\widetilde{\mu}_2} \norm{E_2-\widetilde{E}_2} .
	\end{aligned}
	\end{equation}
	Note that, by the proof of the Theorem \ref{Theo.Estim}, the difference $|\widetilde{\alpha_2}-\alpha_2|$ is controlled by $\varepsilon_S$ and $\varepsilon_E$. And the difference $|\gamma_2+\mu_2 -\widetilde{\gamma}_2-\widetilde{\mu}_2|$ is controlled by $\varepsilon_R$ and $\varepsilon_D$.
	
	Finally, for the class of recovered individuals, $\norm{R_2-\widetilde{R}_2}$, we can use Remark \ref{Remark.R} to obtain 
	\begin{equation}\label{Eq.EstimateR.better}
	\norm{R_2-\widetilde{R}_2} \leq \K_R(T) := \left| \frac{\gamma_2}{\mu_2}D_2(T)  -  \frac{\widetilde{\gamma}_2}{\widetilde{\mu}_2}\widetilde{D}_2(T) +\left(\frac{\widetilde{\gamma}_2}{\widetilde{\mu}_2}-\frac{\gamma_2}{\mu_2}\right) D_2(0)\right| .
	\end{equation}
\end{remark}

\begin{remark}\label{Remark.ErrorC}
Consider the cumulative number of infectious individuals in both groups for the first set of parameters
\begin{equation}\label{Eq.CumulativeFunctionsRef}
C_1(t):=I_1(t)+R_1(t)+D_1(t)\quad\mbox{ and }\quad C_2(t):=I_2(t)+R_2(t)+D_2(t).
\end{equation}
and the respective functions for the second set of parameters
\begin{equation}
\widetilde{C}_1(t)=\widetilde{I}_1(t)+\widetilde{R}_1(t)+\widetilde{D}_1(t)\quad\mbox{ and }\quad \widetilde{C}_2(t):=\widetilde{I}_2(t)+\widetilde{R}_2(t)+\widetilde{D}_2(t) .
\end{equation}
Estimating the difference $(S_2+E_2)(t)-(\widetilde{S}_2+\widetilde{E}_2)(t)$ as we did in Step 2 of the proof of Theorem \ref{Theo.Estim}, and using \eqref{Eq.NumIndividuals}, we obtain
\begin{equation}\label{Eq.EstimateC2}
\norm{C_2-\widetilde{C}_2}= \norm{S_2+E_2-\widetilde{S}_2-\widetilde{E}_2}\leq K_E(T) .
\end{equation}
\end{remark}

Notice that hypothesis \eqref{Hyp.i} is related to social interaction between the two groups, since it involves $\beta_{ij}$. On the other hand the hypotheses \eqref{Hyp.ii}, \eqref{Hyp.iii} and \eqref{Hyp.iv} are on the parameters $\alpha_i$, $\gamma_i$ and $\mu_i$, which depend on {biological features of the disease, and also on how each group deals with the infected individuals, such as access to the health system and access to medicines.Thus, these conditions may play a role whether the definition of groups 1 and 2 refer to social or biological features.}


\section{Examples with the COVID-19} \label{Section.COVID}
{In this section we propose a method for applying the theoretical results presented in Section \ref{Section.Estimates}. The method is developed in the Subsection \ref{SubSection.NY}  and applied to New York State early epidemics of COVID-19 for validation of its predictions, since in this case data is fully known. In Subsection \ref{SubSection.PetrJua} we apply the same method to the COVID-19 epidemics in the cities of Juazeiro and Petrolina, which are neighboring cities in the northeast of Brazil that face different health policies. Our goal with this example is to detect a potential problem in the number of reported cases, since the city of Juazeiro is under a different health system and applied proportionally much less tests for detecting infected individuals compared to the city of Petrolina.}

\subsection{New York County}\label{SubSection.NY}
In this subsection, we show a numerical simulation of the results obtained in Section \ref{Section.Estimates} based on data for the early stage of the COVID-19 epidemic in the New York County and neighboring counties. We will illustrate the estimates in Theorem \ref{Theo.Estim} using data for the number of cases and deaths in the neighboring counties and deaths in the New York County to estimate the curve of infection in the New York County population {and an error range for this estimation}. The choice of these populations was made only by geographic proximity.

We consider as Group 1 the inhabitants of the following counties: Richmond, Kings, Queens, Bronx, Monmouth, Hudson and Bergen. And as Group 2 the inhabitants of New York County.  Due to the restrictions imposed by the ``New York State on Pause'' Executive Order, which affected social distancing and economic activities from late March 2020 onwards, and thus changed the interaction rates between groups, we restrict the analysis to the first 30 days of the outbreak of the disease: from March 5 (first day with registered cases in both groups) till April 3. Therefore, we will use the data presented in Table \ref{Tabela.Dados}.

\begin{table}[h!]
	\resizebox{\textwidth}{!}{	\begin{tabular}{ | c | c | c | c | c | c | c | c| c | c | c | c | c |}
			\hline 
			Day (Mar)        & 02& 03& 04& 05& 06&07 &08 &09&10&11&12&13\\
			\hline		
			Cases Group 1 &0&0&		0&		2&		2&		5&		7&		14&		17&		17&		46&		66\\
			Cases Group 2 &1&1&	1&		3&		4&		8&		8&		11&		17&		18&		39&		48\\
			Deaths Group 1 &0&0&		0&		0&		0&		0&		0&		1&		1&		1&		1&		1\\
			Deaths Group 2 &0&0&		0&		0&		0&		0&		0&		0&		0&		0&		0&		0\\
			\hline
			\hline 
			Day (Mar) &14&15&16&17&18&19&20&21&22&23&24&25\\
			\hline		
			Cases Group 1 &109&151&		246&		475&		1017&		1833&		2449&		3887&		4877&		6497&		7925&		11364\\
			Cases Group 2 &71&72&		111&		277&		590&		1038&		1314&		1863&		2072&		2572&		2887&		3616\\
			Deaths Group 1 &1&2&		4&		4&		4&		6&		7&		10&		46&		81&		125&		173\\
			Deaths Group 2 &0&0&		0&		0&		0&		0&		0&		0&		10&		19&		35&		43\\
			\hline
			\hline 
			Day (Mar-Apr)  & 26& 27&28&29&30&31&01&02&03&04&05&06\\
			\hline		
			Cases Group 1 &13764&15613&		18986&		21014&		24526&		27382&		30575&		34036&		39407&		43801&		47325&		51323
			\\
			Cases Group 2 &4046&4478&		5237&		5582&		6060&		6539&		7022&		7398&	8452&		9300&		9898&		10440	
			\\
			Deaths Group 1 &229&287&		418&		491&		570&		712&		935&		1123&		1310&		1573&		1775&		2012	 	
			\\
			Deaths Group 2 &55&65&		93&		103&		119&		129&		165&		178&		215&		264&		390&	436
			\\
			\hline
	\end{tabular}}
	\caption{Cumulative daily reported cases and deaths confirmed by testing from March 3, 2020 to April 5, 2020, for New York County (Group 2), and Richmond, Kings, Queens, Bronx, Monmouth, Hudson and Bergen counties (Group 1)  according to \cite{USAfacts}.}
	\label{Tabela.Dados}
\end{table}

 {In order to illustrate Theorem \ref{Theo.Estim}, we are assuming that data in Table \ref{Tabela.Dados} comes from a model \eqref{Eq.SEIRD} (possibly with noises) for some unknown set of parameters. This set of unknown parameters will be used as the first set \orange{mentioned} in Theorem \ref{Theo.Estim} and will be represented as $\widetilde{\A}$, $\widetilde{\B}$, $\widetilde{\Gamma}$, $\widetilde{\M}$. Also, the respective solutions to the model will be indicated by tilde. Therefore, from the data we will obtain the (assumed) model solutions $\widetilde{D}_1$ and $\widetilde{D}_2$, and the cumulative function $\widetilde{C}_1:=\widetilde{S}_1+\widetilde{E}_1+\widetilde{I}_1$ (we give more details later). The compartments $\widetilde{S}_1, \widetilde{E}_1$ and $\widetilde{R}_1$ will be extrapolated from these two. The solutions $\widetilde{S}_2, \widetilde{E}_2, \widetilde{I}_2$ and $\widetilde{R}_2$ are assumed to be unknown. We can think of it as if the data for group 2 is full of measurement errors and thus unreliable for the number of cases, but reliable for number of deaths. Thus, we will use the data on the evolution of cases of infection and deaths in Group 1 and also the evolution of deaths in Group 2 to estimate the parameters for the model. These estimated parameters will be used as the second set of parameters \orange{mentioned} in Theorem \ref{Theo.Estim} and will be represented as $\A, \B $, $\Gamma$, $\M$, without tilde, with the respective solutions represented also without the tilde. For the group 1, we will have the solutions obtained from the data (with the tilde), since the parameters are unknown, and the solutions obtained from the model with estimated parameters (without the tilde), and thus we can measure the distances between them, as presented in \eqref{Eq.group1}. For the group 2, the same can be done, only for the compartment $D$, \orange{since the only reliable data for group 2 are the number of deaths}, and thus will be measured as in \eqref{Hyp.v} of Theorem \ref{Theo.Estim}. Therefore, we will use Theorem \ref{Theo.Estim} to estimate the distance between the solutions for the other compartments of \orange{group 2}. Thus, we will be able to measure how far the solutions for the estimated model (without the tilde) are from the unknown solutions (with tilde). The theorem gives a measure on \orange{how the curves obtained by plotting the solutions of the estimated} model accurately depict the evolution of the cases in Group 2, since it gives us an error estimate where the actual curve must be.}
	
	

The hypotheses of Theorem \ref{Theo.Estim} are imposed in the fitting process as follows. Firstly, for the hypothesis \eqref{Hyp.ii} we will suppose that the latency rate is the same for both groups, i.e. $\alpha_1=\alpha_2=\alpha$, since this is a parameter related exclusively to the disease itself. Furthermore, by the same reason, after fitting this parameter to the dataset, we will consider only solutions with this same latency. This means that we will consider also $\widetilde{\alpha}_1=\widetilde{\alpha}_2=\alpha$. {We will make the same assumption with respect to the hypothesis \ref{Hyp.iii}. Therefore we will consider that $\widetilde{\gamma}_1=\widetilde{\gamma}_2=\gamma=\gamma_1=\gamma_2$.}  {We will not impose  $\mu_1=\mu_2$ because the mortality rates may not depend only on the disease itself, but be influenced by the quality of medical treatment and access to it in each group.}

Let us clarify how the data from Table \ref{Tabela.Dados} will be used. Let $\widetilde{C}_1$ and $\widetilde{D}_1$ be the  {7 days centered moving average of the cumulative} number of cases and  {cumulative} number of deaths in Group 1, respectively, as given in Table \ref{Tabela.Dados}. And $\widetilde{C}_2$ and $\widetilde{D}_2$ be the  {7 days centered moving average of the cumulative} number of cases and  {cumulative number of} deaths in Group 2, respectively as given in Table \ref{Tabela.Dados}.  Thus, in order to fit the parameters for the SEIRD model, \orange{we will use the values} of $\widetilde{C}_1$, $\widetilde{D}_1$ and $\widetilde{D}_2$ only. After fitting them, we will use the solution of the fitted SEIRD model for the number of cases in Group 2 and Theorem \ref{Theo.Estim} to estimate $\widetilde{C}_2$.

From the functions defined above, we can define the approximated number of susceptible individuals in each group by
\begin{equation}\label{Eq.Stil}
	\widetilde{S}_i(t)=N_i-\frac{\widetilde{C}_i(t+1)-\widetilde{C}_i(t)}{\alpha}-\widetilde{C}_i(t) ,\qquad i=1,2
\end{equation}
since the second term on the right-hand side approximates \orange{the} number of exposed individuals.

For the hypothesis \eqref{Hyp.i}, we will estimate the constants $V_1=(\widetilde{\beta}_{21}-\widetilde{\beta}_{11})/\widetilde{\mu}_1$ and $V_2=(\widetilde{\beta}_{22}-\widetilde{\beta}_{12})/\widetilde{\mu}_2$ as follows: from Lemma \ref{S.ratio}, we have

\[\displaystyle \frac{\widetilde{S}_1(t)}{\widetilde{S}_1(0)}= \frac{\widetilde{S}_2(t)}{\widetilde{S}_2(0)}\exp\left\{ \int_0^t\left(\left(\widetilde{\beta}_{21}-\widetilde{\beta}_{11}\right)\widetilde{I}_{1}(s)+\left(\widetilde{\beta}_{22}-\widetilde{\beta}_{12}\right)\widetilde{I}_{2}(s)\right)\,ds \right\}.
\]

{Since we are not using the data on the cases of infection in Group 2, we need to estimate $\widetilde{S}_2(t)/\widetilde{S}_2(0)$. Note that for small $t$ we have $\widetilde{S}_2(t)/\widetilde{S}_2(0)\approx 1$, by continuity}. Therefore, using the equations for $\widetilde{S}_1$ and $\widetilde{S}_2$, for small $t$ we have
\begin{equation}
\frac{\widetilde{S}_1(t)}{\widetilde{S}_1(0)}\approx\exp\left\{ V_1(\widetilde{D}_{1}(t)-\widetilde{D}_1(0))+V_2(\widetilde{D}_{2}(t)-\widetilde{D}_2(0))\,ds \right\}. \label{apl.ratio}
\end{equation}

For a pair of time values $t=t_1$ and $t=t_2$, the equation (\ref{apl.ratio}) generates the  linear system, that has solutions:
\begin{eqnarray}
V_1\approx\displaystyle\frac{\widetilde{D}_2(t_2)\ln\left(\frac{\widetilde{S}_1(t_1)}{\widetilde{S}_1(0)}\right)-\widetilde{D}_2(t_1)\ln\left(\frac{\widetilde{S}_1(t_2)}{\widetilde{S}_1(0)}\right)}{\widetilde{D}_1(t_1)\widetilde{D}_2(t_2)-\widetilde{D}_1(t_2)\widetilde{D}_2(t_1)};\label{V.values1}\\ V_2\approx\displaystyle\frac{\widetilde{D}_1(t_1)\ln\left(\frac{\widetilde{S}_1(t_2)}{\widetilde{S}_1(0)}\right)-\widetilde{D}_1(t_2)\ln\left(\frac{\widetilde{S}_1(t_1)}{\widetilde{S}_1(0)}\right)}{\widetilde{D}_1(t_1)\widetilde{D}_2(t_2)-\widetilde{D}_1(t_2)\widetilde{D}_2(t_1)},\label{V.values2}
\end{eqnarray}
as long as $\widetilde{D}_1(t_1)\widetilde{D}_2(t_2)-\widetilde{D}_1(t_2)\widetilde{D}_2(t_1)\not=0$.  The quotients $\widetilde{S}_1(t_j)/\widetilde{S}_1(0)$, $j=1,2$, can be approximated from the data using \eqref{Eq.Stil}. 

Therefore, one can calculate $V_1$ and $V_2$ for any pair of time values $t=t_1$ and $t=t_2$. We took the average of the respective \orange{absolute values of} $V_1$ and $V_2$ for all pairs of integer times $t$ between $t=1$ and $t=30$ days \orange{and then set the signs such that $V_1<0$ and $V_2>0$}. Thus we obtain {$V_1=-3.023433\times 10^{-7}$} and  {$V_2 =1.991887\times 10^{-5}$}.  These values are used to fix the relation between the transmission rates $\beta_{ij}$ and the mortality rates, due to hypothesis \eqref{Hyp.i}. Hence, given $\beta_{i1}$ and $\mu_i$, the corresponding value for $\beta_{i2}$ is given by $\beta_{i2}=\beta_{i1}-\mu_iV_i$. \orange{We observed that taking the average of \eqref{V.values1}-\eqref{V.values2} directly led to unrealistic values of $V_i$. 
Furthermore, using the absolute values to estimate $V_1$ and $V_2$ and setting them to have opposite signs gave best adjustments to the data. That is compatible with the fact that the local transmissions are stronger than the infection between individuals of different groups.}

We consider the following set of parameters: $\beta_{11},$ $\beta_{12}$, $\beta_{21},$ $\beta_{22},$ $ \alpha$, $\gamma_1$, $\gamma_2$, $\mu_1$, $\mu_2$, $E_1(0)$, $E_2(0)$. We are considering the initial conditions as a parameter too because it is not clear from the dataset how to obtain a good approximation for the evolution of the exposed class. Thus, we will use the initial conditions for the number of susceptible individuals in each group as $S_1(0)=N_1-I_1(0)-E_1(0)$ and $S_2(0)=N_2-I_2(0)-E_2(0)$, where {$N_1= 0.84 N$ and $N_2=0.16 N$, with $N=10.5\times 10^6$} being the total population of Group 1 and 2 together  {(see \cite{PopulationUSA})}. 

Under the above conditions, we fitted all the parameters mentioned above together using the Weighted Root Mean Square Error (WRMSE). The weights were used to compensate for the fact that the first group is much bigger. \orange{Therefore, we used weight 1 for the cases in Group 1, and weights 40 and 240 for deaths in Group 1 and Group 2, respectively}. The numeric calculations were made by the software R-CRAN, using the packages deSolve (\cite{R-DE}) and DEoptim (\cite{DEoptim}) for solving the ODEs and {optimizing all the parameters together}, respectively. In order to reduce overfitting, we have already set $\alpha_1=\alpha_2=\alpha$, $\gamma_1=\gamma_2$, and fixed the relation between $\beta_{1i},\beta_{2i}$ and $\mu_i$ through $V_i$, $i=1,2$, but we also impose  {$\alpha=0.3$ (value chosen from the range observed in \cite{WHO})} and a restriction to the optimizing process by considering only parameters satisfying $\beta_{i,j}\in (0,10/N)$, $i,j=1,2$, and $\gamma\in(1/14,1/4)$. These bounds are based on the observed values in \cite{WHO,Lauer,Lin,Wang}. The Table \ref{Tabela.ParametrosAjuste2} presents the fitted parameters for the SEIRD system \eqref{Eq.SEIRDcomponentwise}.

\begin{table}[h!]
	\centering
	\begin{tabular}{ | c  c ||  c  c || c  c| }
		\hline 
		\multicolumn{6}{|c|}{Parameters} \\
		\hline
		\multicolumn{1}{|c|}{Symbol}  & Value & \multicolumn{1}{|c|}{Symbol}& Value & \multicolumn{1}{|c|}{Symbol}& Value \\
		\hline
		$\beta_{11}$ & $5.817\times 10^{-8}$ & $\mu_1$  &0.010349  & $E_1(0)$  &94.64  \\
		$\beta_{12}$   & 	$5.609\times 10^{-8}$  & $\mu_2$&  0.006640     & $E_2(0)$ &81.69 \\
		$\beta_{21}$&  $	5.504\times 10^{-8}$  & $\gamma$ &0.192012   & &   \\
		$\beta_{22}$ & $1.883\times 10^{-7}$  & $\alpha$ &0.3    &	&\\
		\hline
	\end{tabular}

	\caption{Fitted parameters for the system \eqref{Eq.SEIRDcomponentwise} for the data of reported cases and deaths in Group 1 and reported deaths in Group 2 as presented in Table \ref{Tabela.Dados}. {The measurement units of the parameters are: $\beta_{ij}$ is in $(\mbox{individuals} \cdot \mbox{days})^{-1}$, $\mu_i, \gamma$ and $\alpha$ are in $(\mbox{days})^{-1}$, and $E_i(0)$ are in individuals, for all $i,j\in\{1,2\}$.}}
	\label{Tabela.ParametrosAjuste2}
	\vspace{-0.8\baselineskip}
\end{table}

Let $S_1,S_2,E_1,E_2,I_1,I_2,R_1,R_2,D_1$ and $D_2$ the solutions to \eqref{Eq.SEIRD} with the parameters showed in Table \ref{Tabela.ParametrosAjuste2}. Then the function of cumulative number of infections individuals for both groups is
\begin{equation}\label{Eq.SolucaoReferencia}
C_1(t):=I_1(t)+R_1(t)+D_1(t)\quad\mbox{ and }\quad C_2(t):=I_2(t)+R_2(t)+D_2(t).
\end{equation}

The Figure \ref{Fig.Ajustepor1} shows the graphs of $C_1(t)$, $C_2(t)$, $D_1(t)$ and $D_2(t)$ with the respective cumulative number of reported cases and deaths obtained from the dataset.

\begin{figure}[h!]
	\begin{subfigure}{.495\textwidth}
		\caption{Infectious cases in Group 1}
		\includegraphics[width=\linewidth]{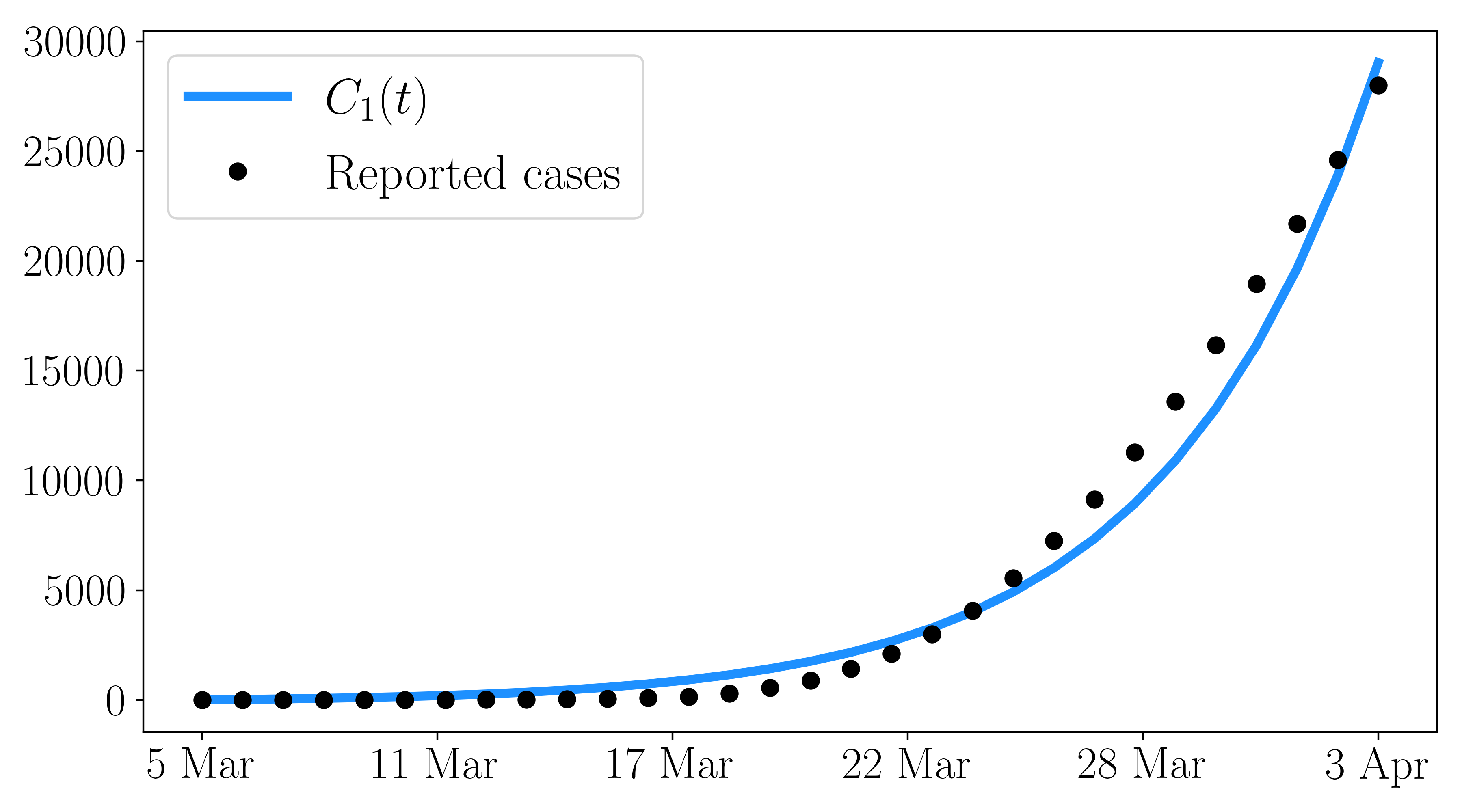}  
		\label{fig.sub.CasosViz}
	\end{subfigure}
	\begin{subfigure}{.495\textwidth}
		\caption{Deaths in Group 1}
		\includegraphics[width=\linewidth]{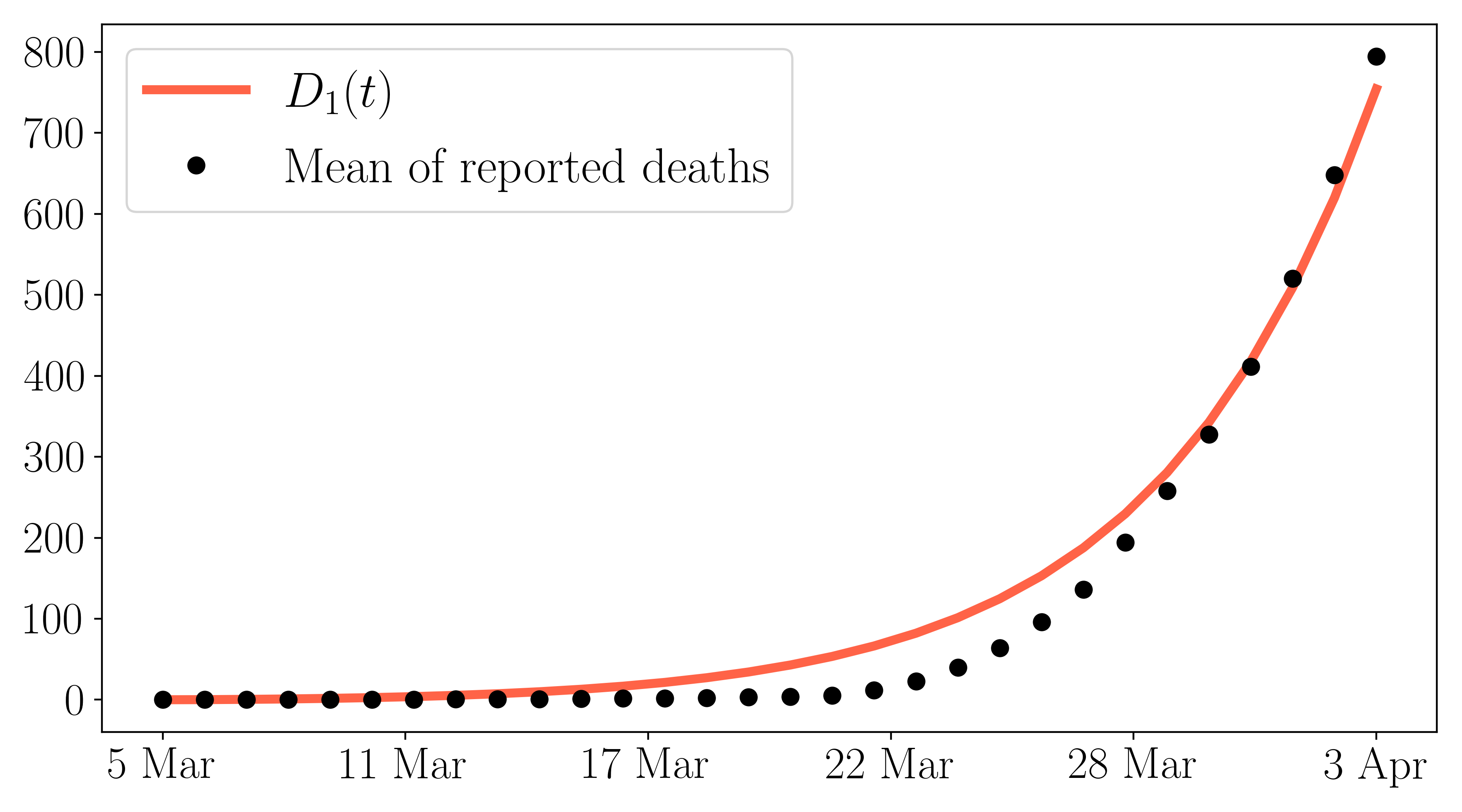}  
		\label{fig.sub.MortesViz}
	\end{subfigure}
	\vspace{-0.5cm}
	\newline	
	\begin{subfigure}{.495\textwidth}
		\caption{Infectious cases in Group 2}
		\includegraphics[width=\linewidth]{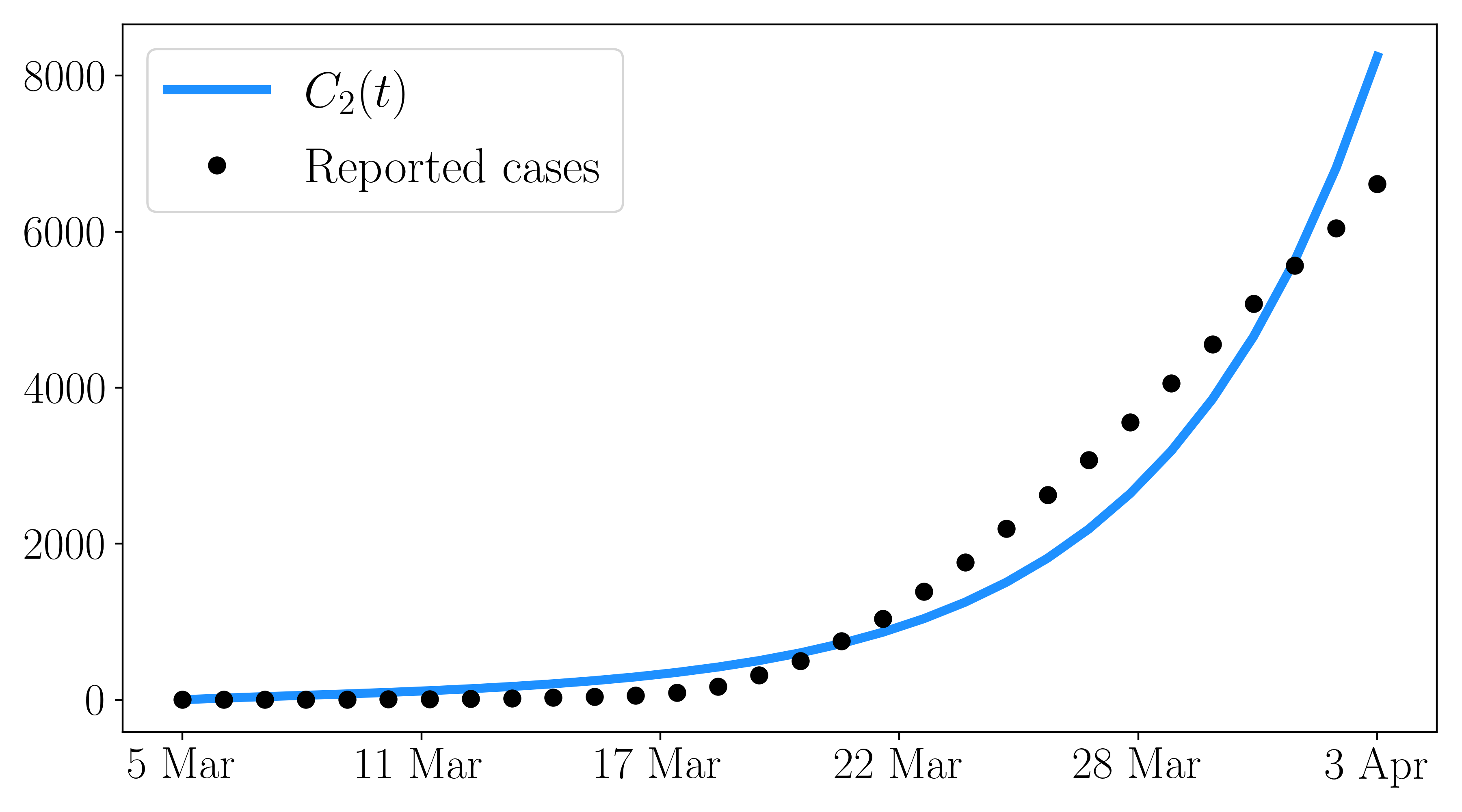}  
		\label{fig.sub.CasosNY}
	\end{subfigure}
	\begin{subfigure}{.495\textwidth}
		\caption{Deaths in Group 2}
		\includegraphics[width=\linewidth]{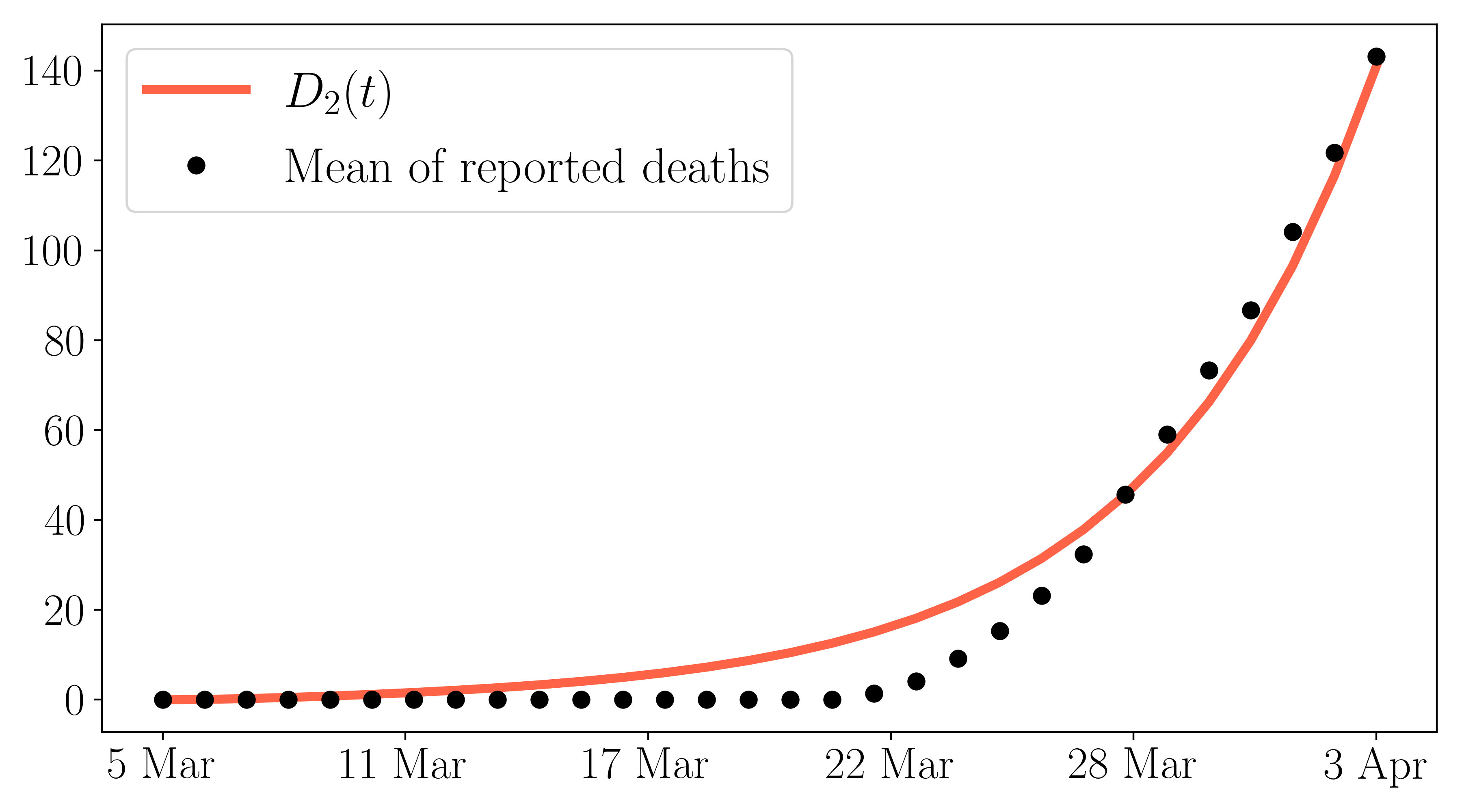}  
		\label{fig.sub.MortesNY}
	\end{subfigure}
	\vspace{-1.5\baselineskip}
	\caption{Fit of the SEIRD system \eqref{Eq.SEIRDcomponentwise} to the reported data. In (\subref{fig.sub.CasosViz}) we plot the cumulative number of reported infectious cases (black dots) and the function $C_1(t)$ (blue solid line) for the Group 1 (neighboring counties). In (\subref{fig.sub.MortesViz}), the reported deaths (black dots) and the function $D_1(t)$ (red solid line) for the Group 1. In (\subref{fig.sub.CasosNY}) we plot the cumulative number of reported infectious cases (black dots) and the function $C_2(t)$ (blue solid line) for the Group 2 (New York County). In figure (\subref{fig.sub.MortesNY}), the reported deaths (black dots) and the function $D_2(t)$ (red solid line) for the Group 2.}
	\label{Fig.Ajustepor1}
\end{figure}

In order to obtain the estimate in Theorem \ref{Theo.Estim}, let us define for every $t=1,...,30$ the distances
\[
\epsilon_{D_1}(t)=\max_{s=1,\ldots,t}|D_1(s)-\widetilde{D}_1(s)|\;,\;\; \epsilon_{D_2}(t)=\max_{s=1,\ldots,t}|D_2(s)-\widetilde{D}_2(s)|
\]
and 
\[\epsilon_{S_1}(t)=\max_{s=1,\ldots,t}|S_1(s)-\widetilde{S}_1(s)|\;,\]
where $\widetilde{D}_1, \widetilde{D}_2$ and $\widetilde{S}_1$ are the functions defined at the beginning of this example.

Therefore, from Remarks \ref{Remark.better} and \ref{Remark.ErrorC}, we obtain that $\widetilde{C}_2$ must satisfy
\[\max_{s=1,\ldots,t} |C_2(s)-\widetilde{C}_2(s)|\leq K_{C}(t)\]
where
\[K_{C}(t) := \max_{s=1,\ldots,t}\left\{ {S}_2(s)\left(1-\left(1-\frac{\varepsilon_{S_1}(s)}{S_1(s)}\right)e^{-\left(|V_1|\varepsilon_{D_1}(s)+|V_2|\varepsilon_{D_2}(s)\right)}\right)\right\} .
\]

The Figure \ref{Fig.ErroEstimativaNY} presents the functions $C_1(t)$ and $C_2(t)$, the reported cases functions $\widetilde{C}_1(t)$ and $\widetilde{C}_2(t)$ and the range of estimated distance $K_C(t)$.

{For the last day of the period, $t=30$, the values of the estimated distances for the cumulative function of infectious individuals is  $K_C(30)=1713.09$, which represents $20.78\%$ of $C_2(30)$. In  comparison, the fitted curve $C_1(t)$ has an accumulated error from the reported data by 2890.93 cases, which represents  $9.96\%$ of $C_1(30)$.}

\begin{figure}[h]
	\includegraphics[width=\textwidth]{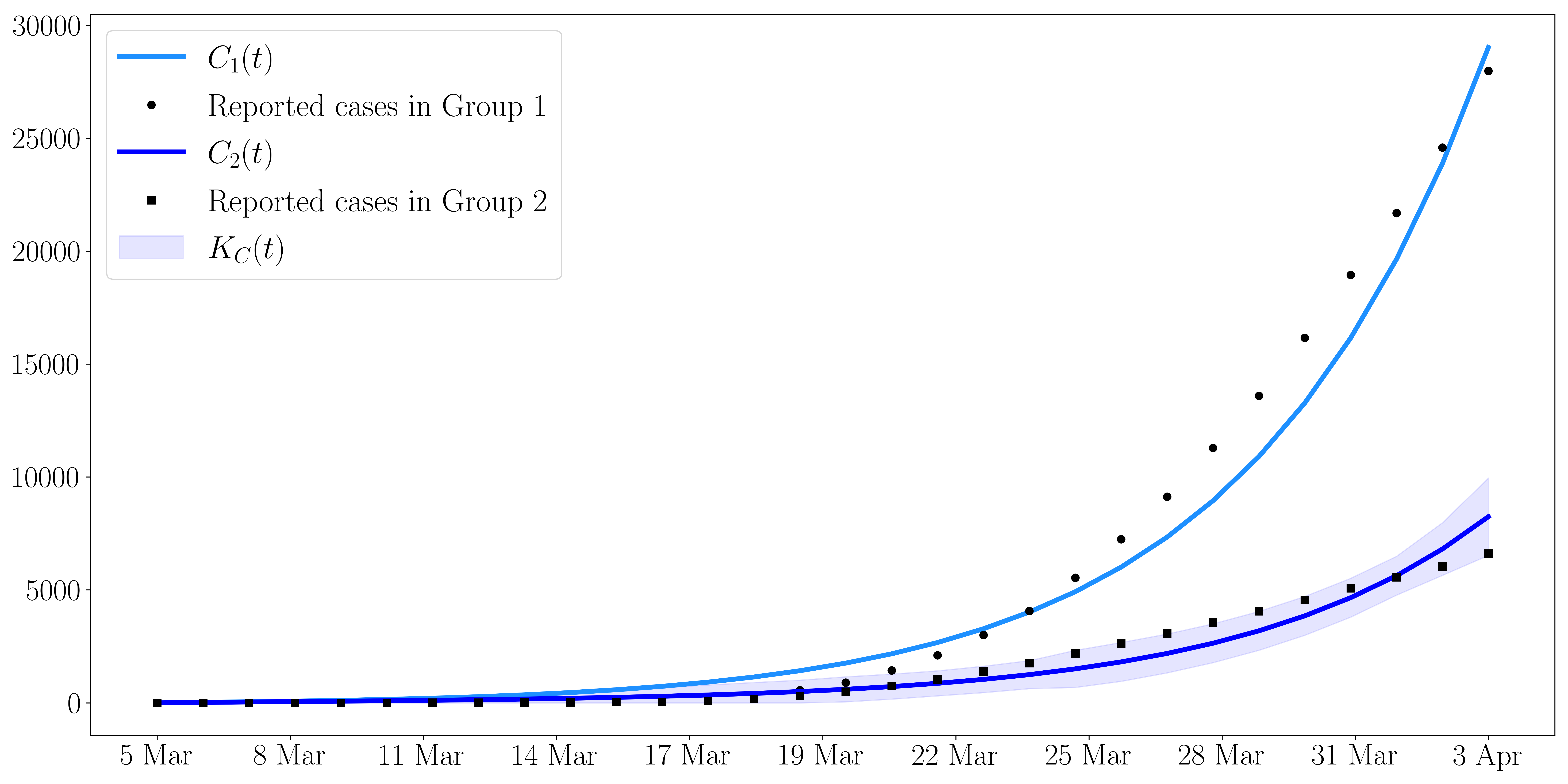}
	\caption{In this figure we plot the functions cumulative number of infectious cases $C_1(t)$ (light blue solid line) and $C_2(t)$ (dark blue solid line), the reported number of infectious cases for Group 1 (circle markers) and for Group 2 (square markers), and the estimated distance range function $K_C(t)$ (blue area).}
	\label{Fig.ErroEstimativaNY}
\end{figure}

It is important to note that although the black square markers (Group 2 data on number of infectious cases) in Figure \ref{Fig.ErroEstimativaNY} were not used in the fitting process for the curves, they lie within the estimated interval predicted by the Theorem \ref{Theo.Estim}. Indeed, since the reported number of infectious cases in Group 2 fall into the estimated distance from $C_2(t)$, this suggests that the evolution of cases follows a two group SEIRD model and the values in \eqref{Hyp.i}-\eqref{Hyp.iv} were well approximated.


\newpage

\subsection{Petrolina and Juazeiro Cities}\label{SubSection.PetrJua}

Petrolina and Juazeiro are two neighboring cities in the northeast of Brazil, with the first one located in the state of Pernambuco, and the latter \orange{located} in the state of Bahia. Although the two populations interact freely, their public health systems are under different administrations, which caused Petrolina to have performed twice as much testing for COVID-19 proportionally as Juazeiro, by May 13 2021, the last date we consider in our estimations. 

In this example, we use the data on COVID-19 spread in Petrolina and the deaths caused by the disease in Juazeiro to estimate the number of cumulative infectious cases in Juazeiro. This estimate and the range of error given by Theorem \ref{Theo.Estim} indicates that Juazeiro's cases might be under-reported. It is worth mentioning that our estimations are made under assumptions \eqref{Hyp.i}-\eqref{Hyp.iv} in Theorem \ref{Theo.Estim} on the parameters and therefore, our conclusions here should not be taken as a precise portrait of the real situation for the considered population. Our aim in this example is to indicate a possible way of estimating under-reporting.

We consider the data of reported cases and deaths for the period of 90 days from February 13 to May 13 of 2021 with Group 1 as the inhabitants of Petrolina, and Group 2 as the inhabitants of Juazeiro. The number of individuals in each group is $N_1=353800$ and $N_2=226200$, respectively, with total population given by $N=580000$. The reader  {can find this dataset} on Petrolina's and Juazeiro's city halls websites, respectively \cite{Petrolina} and \cite{Juazeiro}. To generate the fitted SEIRD model, we proceeded as in the former example and, for  {the} calculated values $V_1=-1.052314\times 10^{-6}$ and $V_2=3.731528\times 10^{-6}$ the fitted parameters for the model \eqref{Eq.SEIRDcomponentwise} can be seen in Table \ref{Tabela.ParametrosPetrJua}.

\begin{table}[h!]
	\centering
	\begin{tabular}{ | c  c ||  c  c || c  c| }
		\hline 
		\multicolumn{6}{|c|}{Parameters} \\
		\hline
		\multicolumn{1}{|c|}{Symbol}  & Value & \multicolumn{1}{|c|}{Symbol}& Value & \multicolumn{1}{|c|}{Symbol}& Value \\
		\hline
		$\beta_{11}$ & $3.364642 \times 10^{-7} $ & $\mu_1$  &0.003744 & $E_1(0)$  &296.46   \\
		$\beta_{12}$   & $3.174329 \times 10^{-7} $  & $\mu_2$& 0.002563     &$E_2(0)$&296.03 \\
		$\beta_{21}$&  $3.325242 \times 10^{-7} $  & $\gamma$ &0.157060   & & \\
		$\beta_{22}$ & $3.269971 \times 10^{-7}$  & $\alpha$ &0.3     &	&   \\
		\hline
	\end{tabular}
	\caption{Fitted parameters for the system \eqref{Eq.SEIRDcomponentwise} for the data of reported cases and deaths in Group 1, presented in \cite{Petrolina}, and reported deaths in Group 2, presented in \cite{Juazeiro}.}
	\label{Tabela.ParametrosPetrJua}
	\vspace{-1.8\baselineskip}
\end{table}

The Figure \ref{Fig.AjustePetrJua} shows the functions $C_1(t)$, $C_2(t)$, $D_1(t)$ and $D_2(t)$ from the system \eqref{Eq.SEIRDcomponentwise} for the set of parameters in Table \ref{Tabela.ParametrosPetrJua} subjected to the initial conditions
\[\begin{matrix}
S_1(0)=337822.5  , &I_1(0)= 207,&R_1(0)=15284 ,&D_1(0)=190,\\
S_2(0)=217047 , &I_2(0)= 143, &R_2(0)=8554,&D_2(0)=160,
\end{matrix}\]
and the 7-days moving average functions of the total number of reported cases and deaths for each group.

{As shown in Figure \ref{Fig.ErroEstimativaPetrJua}, the cumulative number of reported cases for Group 2 (square markers) stays far below the correspondent curve $C_2(t)$ and out of the estimated distance range given by the blue area. In fact, on the last day of the period $t=90$ (May 13), the model indicates that Group 2 should have a total of $15770$ infectious cases, with an error of $3.26\%$ ($K_C(90)=513.67$), while only $14111$ were reported. This suggests that the reported data for Group 2 are under-reported, or that the hypothesis of the theorem may not be satisfied in this example, in the sense that the estimated values in \eqref{Hyp.i}-\eqref{Hyp.iv} were not well approximated.}

\begin{figure}[h!]
	\begin{subfigure}{.495\textwidth}
		\caption{Infectious cases in Group 1}
		\includegraphics[width=\linewidth]{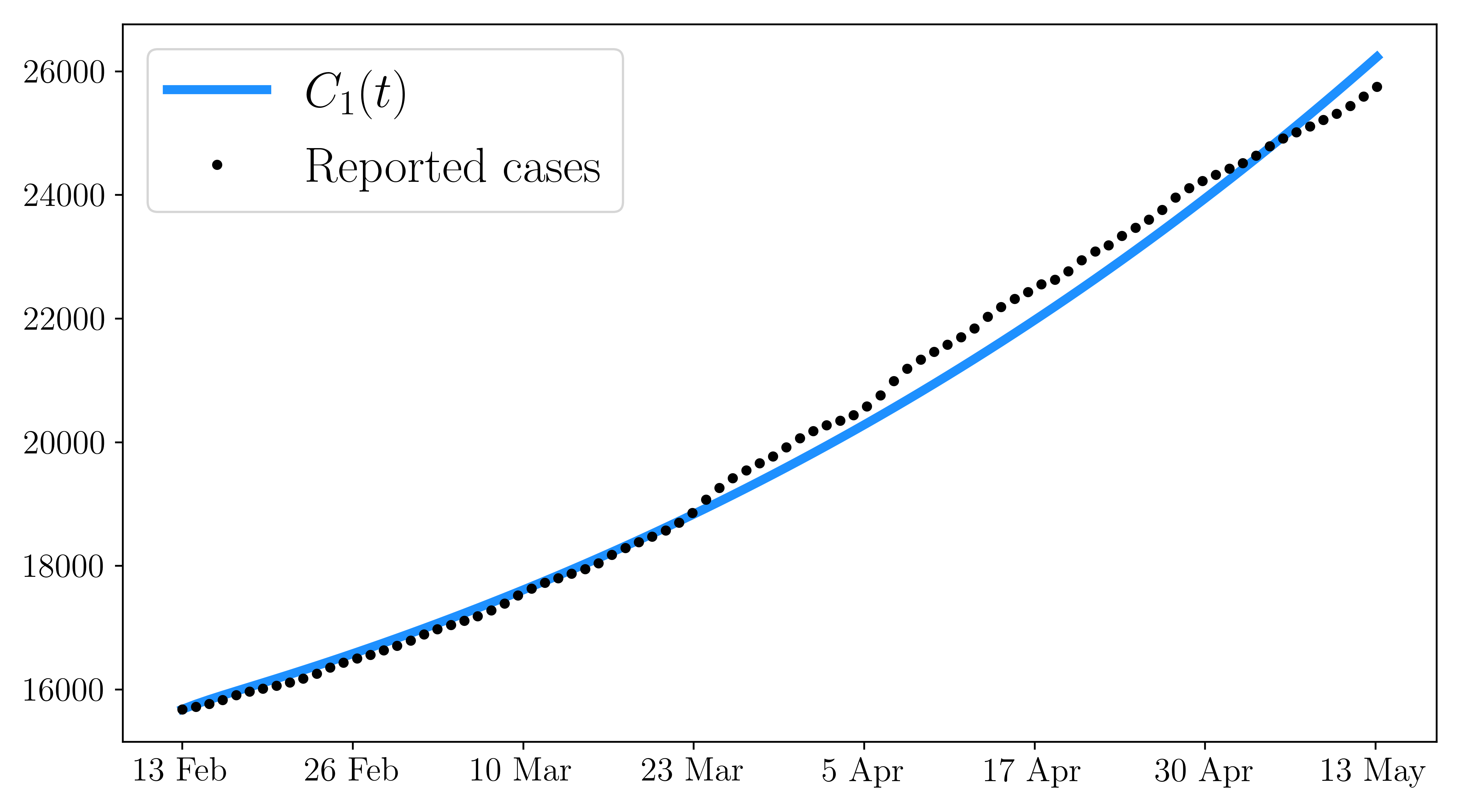}  
		\label{fig.sub.CasosPetr}
	\end{subfigure}
	\begin{subfigure}{.495\textwidth}
		\caption{Deaths in Group 1}
		\includegraphics[width=\linewidth]{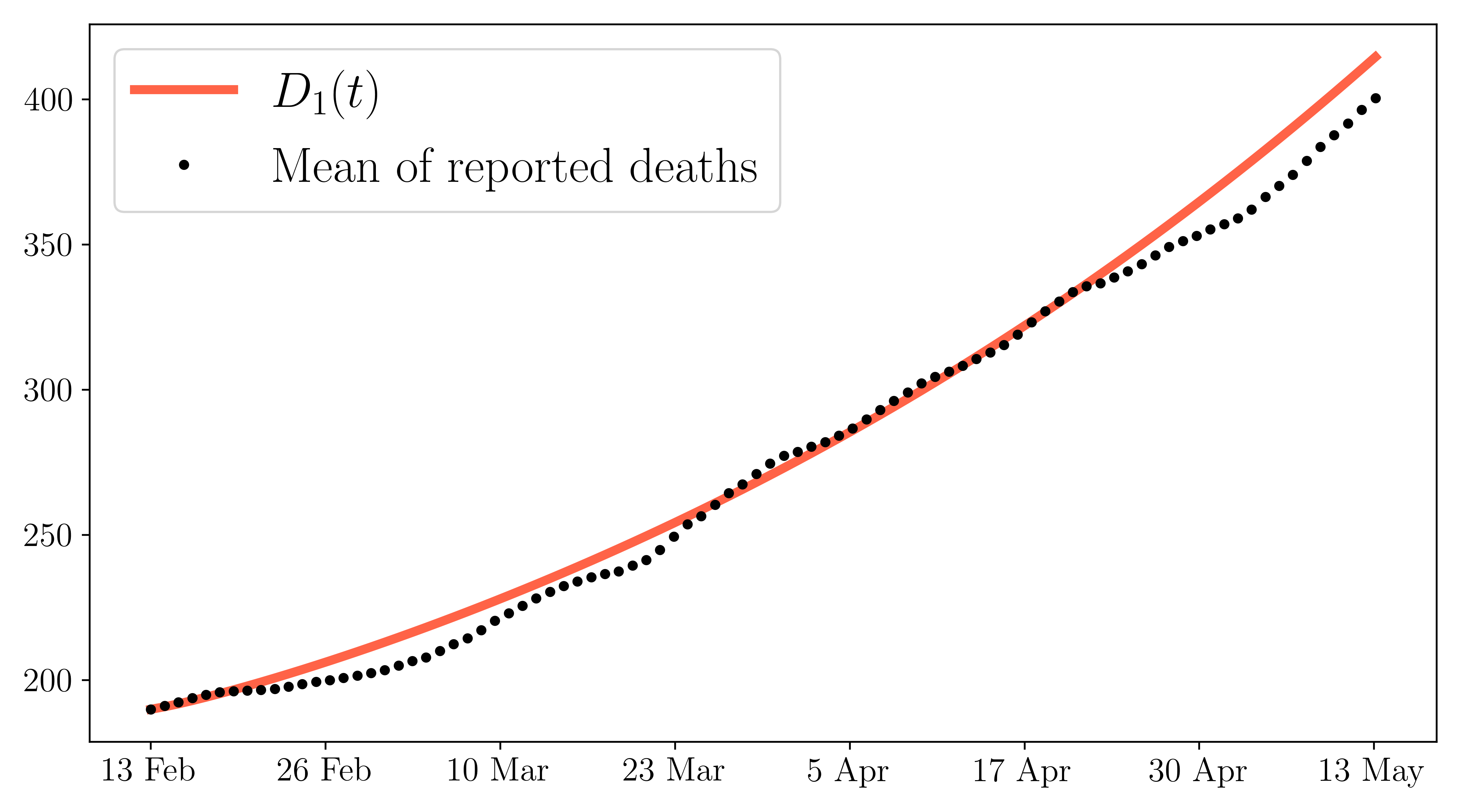}  
		\label{fig.sub.MortesPetr}
	\end{subfigure}
	\vspace{-0.5cm}
	\newline	
	\begin{subfigure}{.495\textwidth}
		\caption{Infectious cases in Group 2}
		\includegraphics[width=\linewidth]{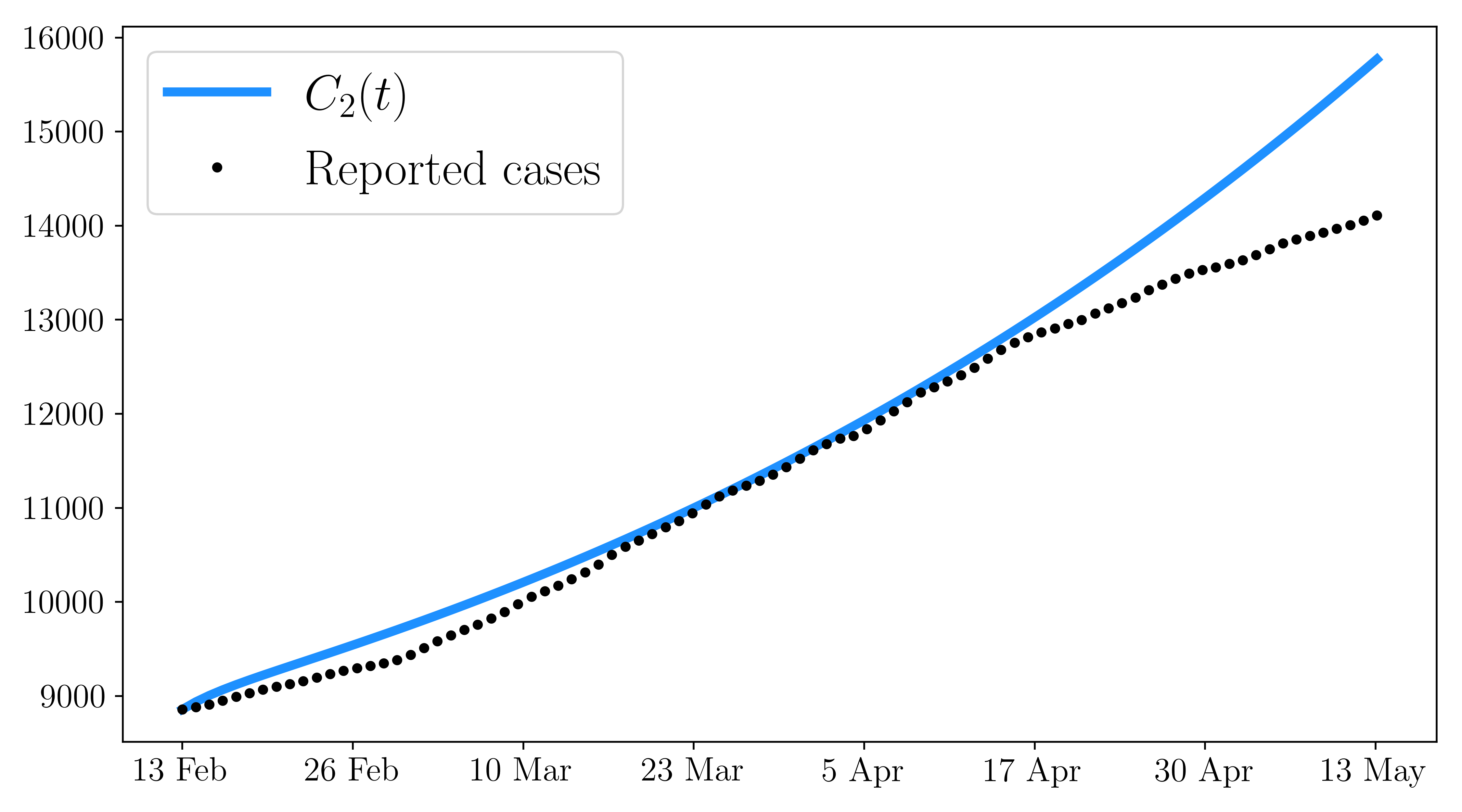}  
		\label{fig.sub.CasosJua}
	\end{subfigure}
	\begin{subfigure}{.495\textwidth}
		\caption{Deaths in Group 2}
		\includegraphics[width=\linewidth]{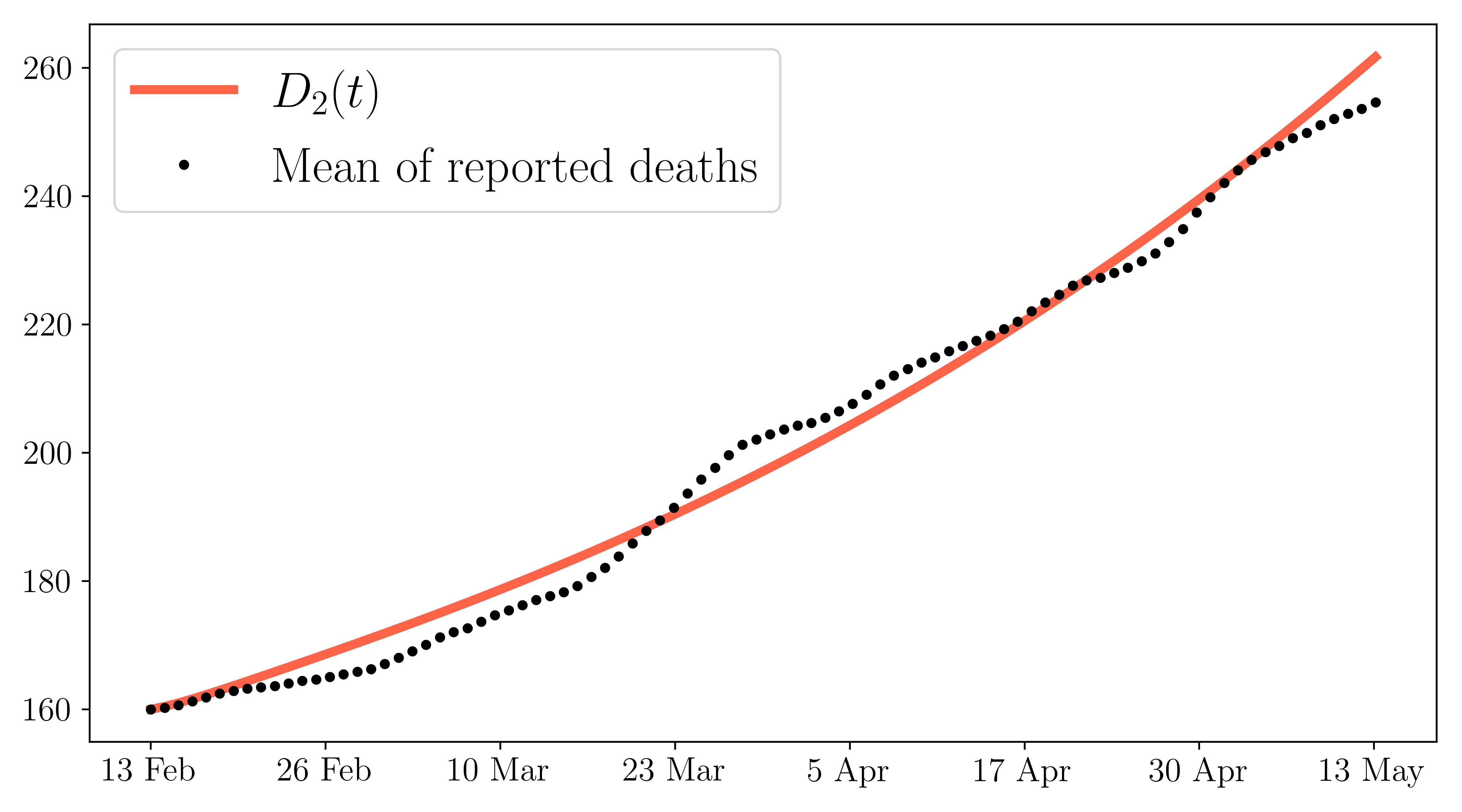}  
		\label{fig.sub.MortesJua}
	\end{subfigure}
	\vspace{-1.5\baselineskip}
	\caption{Fit of the SEIRD system \eqref{Eq.SEIRDcomponentwise} to the reported data. In (\subref{fig.sub.CasosPetr}) we plot the cumulative number of reported infectious cases (black dots) and the function $C_1(t)$ (blue solid line) for the Group 1 (Petrolina). In (\subref{fig.sub.MortesPetr}) we plot the reported deaths (black dots) and the function $D_1(t)$ (red solid line) for the Group 1. In (\subref{fig.sub.CasosJua}) we plot the cumulative number of reported infectious cases (black dots) and the function $C_2(t)$ (blue solid line) for the Group 2 (Juazeiro). In figure (\subref{fig.sub.MortesJua}), the reported deaths (black dots) and the function $D_2(t)$ (red solid line) for the Group 2.}
	\label{Fig.AjustePetrJua}
	\vspace{0.5\baselineskip}
\end{figure}

\begin{figure}[h!]
	\includegraphics[width=\textwidth]{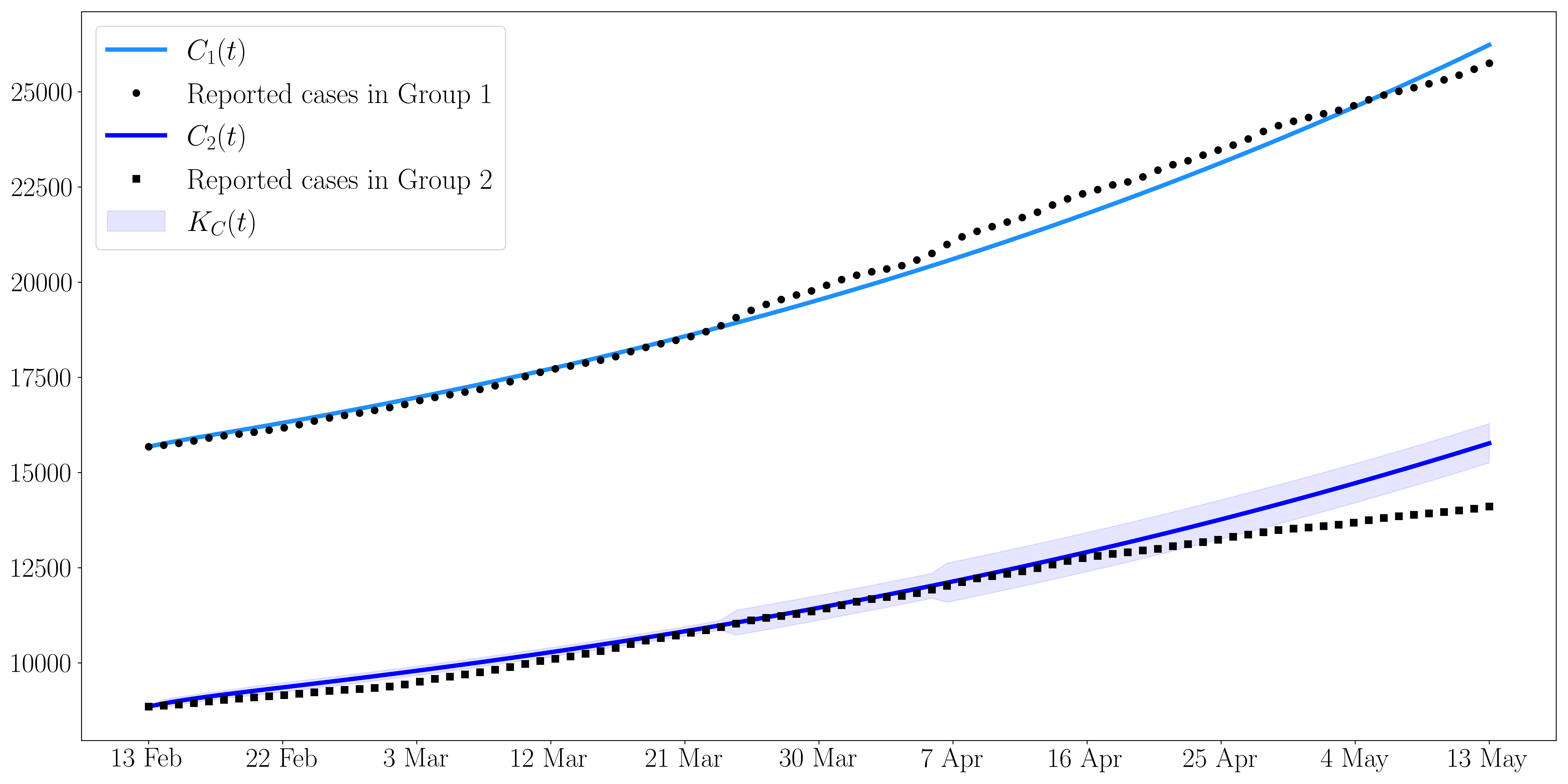}
	\caption{In this figure we plot the function cumulative number of infectious cases $C_1(t)$ (light blue solid line) and the reported number of infectious cases (black circle markers) for Group 1, the function cumulative number of infectious cases $C_2(t)$ (dark blue solid line) and reported number of infectious cases (black square markers) for Group 2, and the estimated distance range function $K_C(t)$ (blue area).}
	\label{Fig.ErroEstimativaPetrJua}
\end{figure}

\newpage

\section{Conclusions and comments}

In the Section \ref{Section.FinalSize} we proved that the final size of the susceptible individuals is always positive for both groups. The result also  {shows} that these values can be obtained by a fixed-point problem in $\R^2$.

The results and discussions presented in Section \ref{Section.Estimates} and \ref{Section.COVID} show that it is possible to use the two group SEIRD model to estimate the dynamics of an epidemics when the number of cases in one of the groups is not known or may be unreliable. 
	
In Section \ref{Section.COVID} we illustrate the results of Theorem \ref{Theo.Estim} with two examples of the spread of COVID-19 disease. The first example was used as a verification case, in order to indicate that the method of estimating the number of cases is consistent. Since the data from the New York State is more reliable, it was expected that the number of reported cases in the New York County would fall into the range of error given by the Theorem with respect to the calculated function of cumulative cases. The parameters obtained for the model SEIRD led to this result and it may be seen in Figure \ref{Fig.ErroEstimativaNY}.  The second application concerns to two neighboring Brazilian cities, Petrolina and Juazeiro, which are under different health policies with Petrolina submitting more inhabitants to testing for COVID-19 than Juazeiro. Therefore, in this example we were interested in checking if the proportionally low number of reported cases in the city of Juazeiro is a potential situation of under-reporting. The implementation of the method for fitting the parameters of the SEIRD model and the range error estimation given by Theorem \ref{Theo.Estim} and showed in Figure \ref{Fig.ErroEstimativaPetrJua}, suggests that the number of infectious cases in Juazeiro may be indeed greater than the ones reported by the health system and present an interval were the number of cases might actually be. \orange{It worth mentioning that we are assuming that reported data from group 1 is reliable and comes from a SEIRD model, possibly with noise, even though it is well known that under-reporting occurred worldwide (\cite{Pul}) and that the data for group 1 may not be accurate, specially in the Petrolina-Juazeiro case.}

	
Since the full dataset is not available, it is not possible to be sure of the results of the method in general. However, combined with other tools, the method might be useful to guide public health polices, for instance when reliable data is lacking for part of the considered population.

In Section \ref{Section.COVID} the definition of the groups is based only on geographic aspects, but one could apply it in situations where the groups are defined by socioeconomic classes, gender or age.
	
It is worth mentioning that in both examples, we are supposing that the data follows a SEIRD model (possibly with noises) for an unknown set of parameters. We then fitted model presented in each case with parameters which satisfy the hypotheses of Theorem \ref{Theo.Estim} with respect to the unknown parameters. In particular, from the available dataset we estimate $V_1=(\widetilde{\beta}_{21}-\widetilde{\beta}_{11})/\widetilde{\mu}_1 ,\ i=1,2$, and therefore in the parameter fitting process we only consider the class parameters for \eqref{Eq.SEIRDcomponentwise} that satisfy these values of $V_i$. I.e., the light blue areas of the Figures \ref{Fig.ErroEstimativaNY} and \ref{Fig.ErroEstimativaPetrJua} contain all solutions to models whose parameters satisfy these values of $V_i$.

We emphasize that these estimations are here only to illustrate the application of our theoretical results. Additional arguments would be needed to validate our assumptions before making assertive statements on the actual scenario of the considered populations.

\vspace{0.4cm}
\noindent {\bf Acknowledgments}
	{The authors are grateful to Sergio Floquet for the fruitful discussions on the subject and the help with the parameter fitting process. The authors would also like to express their sincere thanks to the anonymous referees for their valuable comments and useful suggestions, which significantly contributed to improving the quality of the article.}

\noindent {\bf Conflict of interest} 	
	The authors declare that they have no conflict of interest.

\newpage

\bibliographystyle{myplain}
\bibliography{references}
\end{document}